\documentclass[11pt,a4paper]{article}
\usepackage{todonotes,soul}
\usepackage{amssymb}
\usepackage{amsmath}
\usepackage{amsfonts}
\usepackage{dsfont}
\usepackage{amsthm}
\usepackage{mathrsfs}
\usepackage{hyperref}
\usepackage{color}
\usepackage[margin=2.41cm]{geometry}
\usepackage[all,cmtip]{xy}
\usepackage[utf8]{inputenc}
\usepackage{graphicx}
\usepackage{varwidth}
\usepackage{comment}

\usepackage{upgreek}
\usepackage{rotating}

\usepackage{tikz}
\usetikzlibrary{shapes.geometric}

\usepackage[shortlabels]{enumitem}

\definecolor{darkred}{rgb}{0.8,0.1,0.1}
\hypersetup{
     colorlinks=false,         
     linkcolor=darkred,
     citecolor=blue,
}

\theoremstyle{plain}
\newtheorem{theo}{Theorem}[section]
\newtheorem{lem}[theo]{Lemma}
\newtheorem{propo}[theo]{Proposition}

\theoremstyle{definition}

\newtheorem{remm}[theo]{Remark}
\newenvironment{rem}
  {\pushQED{\qed}\remm}
  {\popQED\endremm}

\numberwithin{equation}{section}

\def\nn{\nonumber}

\def\bbK{\mathbb{K}}
\def\bbR{\mathbb{R}}
\def\bbC{\mathbb{C}}

\def\bbZ{\mathbb{Z}}

\def\bbL{\mathbb{L}}

\def\hom{\underline{\mathrm{hom}}}

\def\Aut{\mathrm{Aut}}

\def\id{\mathrm{id}}

\def\supp{\mathrm{supp}}

\def\dd{\mathrm{d}}
\def\vol{\mathrm{vol}}

\def\cc{\mathrm{c}}

\def\1{\mathbf{1}}

\def\rce{\mathrm{rce}}
\def\RCE{\mathrm{RCE}}

\def\Loc{\mathbf{Loc}}
\def\Lan{\operatorname{Lan}}

\def\AQFT{\mathbf{AQFT}}
\def\Fun{\mathbf{Fun}}

\def\Alg{\mathbf{Alg}}

\def\Vec{\mathbf{Vec}}
\def\Ch{\mathbf{Ch}}
\def\bCh{\mathbf{bCh}}

\def\astdgAlg{\mathbf{dg}^{\ast}\mathbf{Alg}}

\def\dgAlg{\mathbf{dgAlg}}

\def\BB{\mathbf{B}}
\def\CC{\mathbf{C}}
\def\DD{\mathbf{D}}

\def\PoCh{\mathbf{PoCh}}

\def\AAA{\mathfrak{A}}
\def\BBB{\mathfrak{B}}
\def\LLL{\mathfrak{L}}
\def\BBB{\mathfrak{B}}

\def\Sol{\mathfrak{Sol}}
\def\CCR{\mathfrak{CCR}}

\def\O{\mathcal{O}}

\def\colim{\mathrm{colim}}

\def\Tot{\mathrm{Tot}^\oplus}

\def\RCE{\mathrm{RCE}}
\def\rce{\mathrm{rce}}

\newcommand\ovr[1]{\overline{#1}}
\newcommand\mycom[2]{\genfrac{}{}{0pt}{}{#1}{#2}}

\def\sk{\vspace{1mm}}

\makeatletter
\let\@fnsymbol\@alph
\makeatother

\title{%
Relative Cauchy evolution for linear homotopy AQFTs
}

\author{%
Simen Bruinsma$^{1,2,a}$, 
Christopher J.~Fewster$^{2,b}$\ and\ 
Alexander Schenkel$^{1,c}$\vspace{4mm}\\
{\small ${}^1$ School of Mathematical Sciences, University of Nottingham,}\\
{\small University Park, Nottingham NG7 2RD, United Kingdom.}\vspace{2mm}\\
{\small ${}^2$ Department of Mathematics, University of York,}\\
{\small Heslington, York YO10 5DD, United Kingdom.}\vspace{4mm}\\
{\small \begin{tabular}{ll}
Email: & ${}^a$~\texttt{simen.bruinsma@nottingham.ac.uk}\\
& ${}^b$~\texttt{chris.fewster@york.ac.uk}\\
& ${}^c$~\texttt{alexander.schenkel@nottingham.ac.uk}\vspace{2mm}
\end{tabular}
}
}

\date{February 2022}

\begin{document}

\maketitle

\vspace{-5mm}

\begin{abstract}
\noindent This paper develops a concept of relative Cauchy evolution for the class of homotopy algebraic quantum field theories (AQFTs) that are obtained by canonical commutation relation quantiz{\-}ation of Poisson chain complexes.  The key element of the construction is a rectification theorem proving that the homotopy time-slice axiom,  which is a higher categorical relaxation of the time-slice axiom of AQFT,  can be strictified for theories in this class.  The general concept is illustrated through a detailed study of the relative Cauchy evolution for the homotopy AQFT associated with linear Yang-Mills theory, for which the usual stress-energy tensor is recovered.
\end{abstract}

\vspace{-1mm}

\vspace{-2mm}

\vspace{-1mm}

\renewcommand{\baselinestretch}{0.8}\normalsize
\tableofcontents
\renewcommand{\baselinestretch}{1.0}\normalsize



\section{\label{sec:intro}Introduction and summary}
An algebraic quantum field theory (AQFT) on Lorentzian manifolds can be 
described by a functor from a category of spacetimes to a category of algebras,
so that $\AAA(M)$ is the algebra of observables assigned to the spacetime $M$ by the theory.\footnote{It is 
also possible for $\AAA(M)$ to be the field algebra of possibly unobservable fields but for brevity 
we suppress this point.}    
An important concept in this framework is the relative Cauchy evolution (RCE) \cite{BFV,FewsterVerch} 
that allows one to study the response of
a theory $\AAA$ to perturbations of the spacetime geometry. 
It is described by a coherent family of automorphisms
$\RCE_{(M,h)} : \AAA(M)\to\AAA(M)$ of the algebras
of quantum observables that is labeled by
pairs $(M,h)$ consisting of a spacetime $M$ and a sufficiently small 
compactly supported metric perturbation $h$ on $M$.
The first derivatives $\frac{d}{d\epsilon} \RCE_{(M,\epsilon h)} \big\vert_{\epsilon=0}$
of the family of RCE automorphisms along the metric perturbation $h$ 
determine the quantum stress-energy tensor of the theory $\AAA$.
\sk

In the usual construction of the RCE automorphisms,
see e.g.\ \cite{BFV,FewsterVerch}, it is crucial that the theory $\AAA$ satisfies 
the {\em time-slice axiom}. This axiom requires that $\AAA$ assigns to
every \emph{Cauchy morphism} $f:M\to N$ (a morphism of spacetimes
whose image $f(M)\subseteq N$ contains a Cauchy surface of $N$)
an isomorphism $\AAA(f) : \AAA(M)\to\AAA(N)$ of algebras. These isomorphisms introduce a concept 
of time evolution for the quantum observables, which is the key ingredient 
to define the RCE automorphisms. See Section \ref{sec:ordinaryRCE}
for more details and a concise review of this construction.
\sk

The main goal of this paper is to develop a generalization of the concept of 
relative Cauchy evolution to {\em homotopy AQFTs} \cite{BSWhomotopy},
which are higher categorical refinements of AQFTs that are relevant to 
describe quantum gauge theories.  Examples of homotopy AQFTs arise from
the BRST/BV formalism \cite{FredenhagenRejzner,FredenhagenRejzner2}
and its mathematical incarnation in terms of derived geometry \cite{LinearYM}.
See also \cite{CostelloGwilliam} for similar developments in the context 
of factorization algebras.
The key difference between homotopy AQFTs and ordinary ones is that
they assign {\em differential graded algebras} (in short, dg-algebras) of quantum observables
which, in the terminology of the BRST/BV formalism, contain also
observables for the ghost fields and the antifields of a gauge theory.
As evidenced by the concrete examples studied in \cite{FredenhagenRejzner,FredenhagenRejzner2}
and \cite{LinearYM}, physically relevant models of homotopy AQFT usually satisfy
only a higher categorical relaxation of the time-slice axiom, which was
called the {\em homotopy time-slice axiom}. This axiom demands that
the homotopy AQFT $\AAA$ assigns to every Cauchy morphism $f:M\to N$
a {\em weak equivalence} $\AAA(f) : \AAA(M)\to\AAA(N)$ of dg-algebras,
i.e.\ a morphism of dg-algebras whose underlying chain map is a 
quasi-isomorphism (an isomorphism at the level of homologies). In contrast to the isomorphisms in ordinary AQFT, 
such weak equivalences do not admit (strict) inverses and hence the usual
construction of the RCE can not be applied directly to 
homotopy AQFTs. 
\sk

Our proposal to remedy this issue is to seek a {\em rectification theorem}
for the homotopy time-slice axiom of homotopy AQFTs.
Loosely speaking, this consists of a replacement of a given homotopy AQFT $\AAA$ that satisfies
the homotopy time-slice axiom by a weakly equivalent homotopy AQFT $\AAA^{\mathrm{st}}$
that satisfies the strict time-slice axiom and hence admits the usual RCE automorphisms.
The technical details of this proposal are explained in Section \ref{sec:proposalhoRCE}.
The main result is Theorem \ref{theo:rectificationhomotopyAQFT}
which proves that, when restricted to a category that is relevant
for the relative Cauchy evolution associated with a fixed but arbitrary pair $(M,h)$,
every {\em linear} homotopy AQFT satisfying the homotopy time-slice axiom 
admits such a strictification. Let us recall
that a linear homotopy AQFT is a theory $\AAA = \CCR(\LLL,\tau)$ that
is obtained through canonical commutation relation (CCR) quantization
of Poisson chain complexes. In physics terminology,
these are non-interacting (i.e.\ ``free'') quantum gauge theories, such as e.g.\ 
the linear Yang-Mills model studied in \cite{LinearYM}.
Our construction of the weakly equivalent strictified theory 
$\AAA^{\mathrm{st}}$ is rather explicit as it is determined
by CCR-quantization of (a functor of) Poisson chain complexes $(\bbL L_! (\LLL),\tau_\bbL)$ that we 
describe in detail in Sections \ref{sec:rectification} and \ref{sec:Poisson}.
The relative Cauchy evolution for the weakly equivalent strictified theory $\AAA^{\mathrm{st}}$ 
takes a very simple form, given explicitly in \eqref{eqn:bigRCEmap} and \eqref{eqn:smallRCEmap}.
\sk

Even though our concept of relative Cauchy evolution that is 
obtained by the Rectification Theorem \ref{theo:rectificationhomotopyAQFT}
is mathematically sound and relatively simple to describe, its physical 
interpretation is a priori less clear. We shall address this issue
in Section \ref{sec:Example} by studying as a concrete example the linear 
Yang-Mills model from \cite{LinearYM}. The key result is 
Proposition \ref{propo:rcequantumfield} which proves that, at the 
level of linear quantum fields, our abstract concept of relative Cauchy evolution
through rectification admits an equivalent, but physically more transparent and familiar,
description in terms of a generalization of the usual RCE automorphism
construction obtained by choosing quasi-inverses for the quasi-isomorphisms
assigned to Cauchy morphisms. This equivalent model
can be worked out in detail, see in particular
Proposition \ref{propo:rcemapsexplicit} and Remark \ref{rem:rcemapsexplicit},
leading to explicit formulas for the relative Cauchy evolution
of linear Yang-Mills theory that generalize earlier non-homotopical 
results in \cite{FewsterLang}.
We shall further compute the associated stress-energy tensor
and find that, up to exact terms that can be removed by a homotopy, 
it only receives contributions from the gauge field and {\em not} 
from the ghosts and antifields. More precisely, the stress-energy
tensor that is obtained via our homotopical concept of relative 
Cauchy evolution agrees with the one of the standard Maxwell action
$S= -\int_M \frac{1}{2}F\wedge\ast F$. This is a very pleasing 
result because it substantiates mathematically the physical expectation that
ghosts and antifields do not contribute to the energy content of a theory.
\sk

The outline of the remainder of this paper is as follows:
In Section \ref{sec:ordinaryRCE} we shall briefly recall the concept of relative Cauchy evolution
in ordinary AQFT \cite{BFV,FewsterVerch} and provide a reformulation
in terms of localization of categories that will be useful for generalizing to homotopy AQFT.
Section \ref{sec:proposalhoRCE} outlines our proposal for how to obtain 
a well-defined concept of relative Cauchy evolution
for homotopy AQFTs satisfying the homotopy time-slice axiom. The key idea of our approach is
to seek a rectification theorem that allows us to strictify the homotopy time-slice axiom
and thereby make available the usual RCE automorphism construction in the context of homotopy AQFT.
Restricting to linear homotopy AQFTs \cite{LinearYM} and the category \eqref{eqn:RCEcategory}
that is relevant for the relative Cauchy evolution associated with a fixed but arbitrary pair $(M,h)$,
we can successfully prove such a rectification theorem, see Theorem \ref{theo:rectificationhomotopyAQFT}.
The formulation and proof of this theorem is slightly abstract as it requires techniques from model 
category theory \cite{Hovey,DHKS,Riehl}, in particular derived functors and their 
concrete models obtained via bar resolutions \cite{Fresse}. These more technical aspects
are discussed and worked out in detail in Sections \ref{sec:rectification} and \ref{sec:Poisson}.
In Section \ref{sec:Example} we apply our novel concept of relative Cauchy evolution
for homotopy AQFTs to the linear Yang-Mills model from \cite{LinearYM}. We 
obtain explicit formulas involving Green operators, see in particular
Proposition \ref{propo:rcemapsexplicit} and Remark \ref{rem:rcemapsexplicit},
which generalize the earlier non-homotopical results in \cite{FewsterLang}.
We also compute explicitly the stress-energy tensor for this example.
This paper contains three appendices:
Appendix \ref{app:bimcomplexes} summarizes our conventions
for bicomplexes, Appendix \ref{app:bar} recalls the bar construction,
and Appendix \ref{app:zigzag} lists explicitly the homotopy coherence data
that is needed to construct the Poisson structure in Section \ref{sec:Poisson}.


\section{\label{sec:ordinaryRCE}Relative Cauchy evolution and localization of categories}
In this section we shall briefly review the concept of relative Cauchy evolution (RCE) 
for ordinary AQFTs \cite{BFV,FewsterVerch} from a perspective 
that will be useful for our generalization to homotopy AQFTs. 
We adopt the operadic formulation of AQFT \cite{BSWoperad,BSreview} that automatically incorporates 
Einstein causality as an intrinsic part of the structure.
Let us recall that an orthogonal category 
is a pair $\ovr{\CC} := (\CC,\perp)$ consisting of a category $\CC$ and 
a subset $\perp\, \subseteq\mathrm{Mor}(\CC)\, {{}_{\mathsf{t}}\times}{{}_{\mathsf{t}}}\, \mathrm{Mor}(\CC)$
expressing which pairs of morphisms to a common target are considered to be orthogonal. 
A central example for AQFT is $\ovr{\Loc} := (\Loc,\perp_{\Loc}^{})$, where $\Loc$ is
the usual category of all oriented and time-oriented globally hyperbolic Lorentzian manifolds $M$
and morphisms $f:M\to N$ given by orientation and time-orientation preserving isometric embeddings
with open and causally convex image.
The orthogonality relation $\perp_{\Loc}^{}$ is determined by causal disjointness, i.e.\ 
$(f_1:M_1\to N)\perp_{\Loc}^{} (f_2:M_2\to N)$ if and only if the images
$f_1(M_1)$ and $f_2(M_2)$ are causally disjoint open subsets of $N$.
Associated to every orthogonal category $\ovr{\CC}$ 
is a category $\AQFT(\ovr{\CC})$ of AQFTs on $\ovr{\CC}$
with values in the symmetric monoidal category $\Vec_\bbK$ 
of vector spaces over a (fixed) field $\bbK$ of characteristic $0$. 
More precisely, the objects of $\AQFT(\ovr{\CC})$ are given by $\Vec_\bbK$-valued algebras
over the AQFT operad $\O_{\ovr{\CC}}$ associated with the orthogonal category $\ovr{\CC} =(\CC,\perp)$. 
By \cite[Theorem 3.16]{BSWoperad}, each such object admits an equivalent description as a 
functor $\AAA : \CC\to \Alg_{\bbK}$ to the category of associative and unital $\bbK$-algebras
satisfying an abstract version of the Einstein causality axiom encoded by the orthogonality 
relation $\perp$ on $\CC$. For $\ovr{\Loc}$, the category
$\AQFT(\ovr{\Loc})$ describes AQFTs on $\Loc$ in the sense of \cite{BFV}, 
which satisfy the usual Einstein causality axiom but not necessarily the time-slice axiom. 
\sk

Turning to the time-slice axiom, let us recall that a $\Loc$-morphism
$f:M\to N$ is called a \emph{Cauchy morphism} if its image $f(M)\subseteq N$
contains a Cauchy surface of $N$. We denote by $W\subseteq \mathrm{Mor}(\Loc)$
the subset of all Cauchy morphisms in $\Loc$. There are two equivalent ways to implement 
the time-slice axiom. On the one hand, one can consider the full subcategory
$\AQFT(\ovr{\Loc})^{W}\subseteq \AQFT(\ovr{\Loc})$ consisting of all AQFTs $\AAA\in  \AQFT(\ovr{\Loc})$
that assign to every Cauchy morphism $(f: M\to N)\in W$ an isomorphism $\AAA(f) : \AAA(M)\to\AAA(N)$.
On the other hand, one can proceed by \emph{localization}. Informally speaking,
the localization $\Loc[W^{-1}]$ of $\Loc$ at $W$ is a universally constructed 
category in which every Cauchy morphism $f\in W$ possesses an inverse.
More precisely, the localization is characterized by a functor
$L : \Loc\to \Loc[W^{-1}]$ that satisfies the universal property
stated in e.g.\ \cite[Section 7.1]{KashiwaraSchapira}.
The localization functor determines a localization of orthogonal categories 
$\ovr{\Loc}[W^{-1}]:=(\Loc[W^{-1}],L_\ast(\perp_{\Loc}^{}))$,
where $L_\ast(\perp_{\Loc}^{})$ is the minimal orthogonality relation
on $\Loc[W^{-1}]$ containing every $\{(L(f_1),L(f_2))\,:\, f_1 \perp_{\Loc}^{} f_2\}$. 
By construction, the localization functor $L$ defines an orthogonal functor
$L : \ovr{\Loc}\to \ovr{\Loc}[W^{-1}]$, i.e.\ a functor that preserves the orthogonality relations.
With this notation established, the time-slice axiom is automatically implemented 
in every theory belonging to the category $\AQFT(\ovr{\Loc}[W^{-1}])$.
These two perspectives are equivalent due to the following result, which was proven
in \cite[Proposition 4.4]{BSWoperad} and \cite[Proposition 2.21]{BSreview}.
\begin{propo}\label{prop:localizationequivalence}
The orthogonal localization functor $L : \ovr{\Loc}\to \ovr{\Loc}[W^{-1}]$ 
determines an adjunction
\begin{flalign}
\xymatrix@C=3.5em{
L_! \,:\, \AQFT(\ovr{\Loc})~\ar@<0.75ex>[r]&\ar@<0.75ex>[l]  ~\AQFT(\ovr{\Loc}[W^{-1}]) \,:\, L^\ast
}
\end{flalign}
that exhibits $\AQFT(\ovr{\Loc}[W^{-1}])$ as a full reflective subcategory of $\AQFT(\ovr{\Loc})$,
i.e.\ the counit $\epsilon : L_! L^\ast\to \id$ is a natural isomorphism.
This adjunction restricts to an adjoint equivalence
\begin{flalign}
\xymatrix@C=3.5em{
L_! \,:\, \AQFT(\ovr{\Loc})^{W}~\ar@<0.75ex>[r]_-{\sim}&\ar@<0.75ex>[l]  ~\AQFT(\ovr{\Loc}[W^{-1}]) \,:\, L^\ast
}
\end{flalign}
between the two categories  $\AQFT(\ovr{\Loc})^{W}$ and $\AQFT(\ovr{\Loc}[W^{-1}])$,
i.e.\ the restriction of the unit $\eta : \id \to L^\ast L_!$ to 
$\AQFT(\ovr{\Loc})^{W}\subseteq \AQFT(\ovr{\Loc})$ is a natural isomorphism.
\end{propo}

The concept of relative Cauchy evolution is defined for each AQFT $\AAA$ that satisfies the time-slice axiom.
It records the response of $\AAA$ to perturbations of the spacetime metric and thereby encodes
information about the stress-energy tensor of the theory \cite{BFV,FewsterVerch}. In more detail,
given a spacetime $M\in \Loc$ with metric denoted by $g$ and a sufficiently small compactly supported 
metric perturbation $h$
such that the spacetime $M_h$ with the  perturbed metric $g+h$ is also an object in $\Loc$,
one obtains a diagram
\begin{flalign}\label{eqn:RCEdiagram}
\xymatrix@R=1.2em@C=1.2em{
&\ar[dl]_-{i_+}M_+ \ar[dr]^-{j_+}&\\
M&&M_h\\
&\ar[ul]^-{i_-} M_-\ar[ur]_-{j_-}&
}
\end{flalign}
in the category $\Loc$. In this diagram $M_\pm := M\setminus J^{\mp}_{M}(\supp\, h)\in\Loc$
is the spacetime $M$ with the causal past/future of the support of the perturbation $h$ removed
and $i_\pm$, $j_\pm$ are the canonical inclusion morphisms. It is important to observe that
all morphisms in the diagram \eqref{eqn:RCEdiagram} are Cauchy morphisms,
but (except in the trivial case $h=0$) not isomorphisms. 
Given now any $\AAA\in\AQFT(\ovr{\Loc})^{W}$ that satisfies the time-slice axiom,
its application to the morphisms in the diagram \eqref{eqn:RCEdiagram} gives isomorphisms
and one defines the relative Cauchy evolution associated with the pair $(M,h)$ as the automorphism
\begin{flalign}\label{eqn:RCEA}
\RCE_{(M,h)}^{}\,:=\,\AAA(i_-)\,\AAA(j_-)^{-1}\,\AAA(j_+)\,\AAA(i_+)^{-1}\,:\, \AAA(M)~\longrightarrow~\AAA(M)
\end{flalign}
of associative and unital $\bbK$-algebras. The family $\{\RCE_{(M,h)}^{}\}$ of automorphisms for all admissible
pairs $(M,h)$ satisfies by construction various compatibility conditions among its members and also with respect
to $\Loc$-morphisms $f:M\to N$, see e.g.\ \cite[Section 3.4]{FewsterVerchSPASS}.
\sk

Let us now reformulate the concept of relative Cauchy evolution
from the equivalent point of view given by AQFTs on the localized orthogonal category
$\ovr{\Loc}[W^{-1}]$, cf.\ Proposition \ref{prop:localizationequivalence}.
Because the orthogonal localization functor $L : \ovr{\Loc}\to\ovr{\Loc}[W^{-1}]$ 
maps each morphism in the diagram \eqref{eqn:RCEdiagram} to an isomorphism,
we may define the automorphism
\begin{flalign}\label{eqn:rLoclocalized}
r_{(M,h)}^{}\,:=\, L(i_-)\,L(j_-)^{-1}\,L(j_+)\,L(i_+)^{-1}\,:\,L(M)~\longrightarrow~ L(M)
\end{flalign}
in $\ovr{\Loc}[W^{-1}]$. Given now any $\BBB \in \AQFT(\ovr{\Loc}[W^{-1}])$, the relative Cauchy evolution
is defined by applying $\BBB$ to the automorphism \eqref{eqn:rLoclocalized}, i.e.\
\begin{flalign}\label{eqn:RCEB}
\RCE_{(M,h)}^{}\,:=\,\BBB(r_{(M,h)}^{})\,:\, \BBB(L(M))~\longrightarrow \BBB(L(M))\quad.
\end{flalign}
It is clear that the two formulations of the relative Cauchy evolution given 
in \eqref{eqn:RCEA} and \eqref{eqn:RCEB} agree when $\AAA\in\AQFT(\ovr{\Loc})^{W}$ 
and $\BBB\in\AQFT(\ovr{\Loc}[W^{-1}])$ define the same theory, i.e.\ $\AAA = L^\ast \BBB$
with $L^\ast$ the right adjoint in Proposition \ref{prop:localizationequivalence}.
\sk

Summing up, for theories defined on the localized orthogonal category
$\ovr{\Loc}[W^{-1}]$, the relative Cauchy evolution is simply given by the application
of a theory $\BBB\in\AQFT(\ovr{\Loc}[W^{-1}])$ to certain automorphisms \eqref{eqn:rLoclocalized}
in $\ovr{\Loc}[W^{-1}]$. The 
compatibility conditions between $\RCE$ automorphisms
and $\Loc$-morphisms listed in \cite{FewsterVerchSPASS} are then directly encoded in
the localized orthogonal category $\ovr{\Loc}[W^{-1}]$.
\sk

We may pursue the reformulation a bit further by considering the 
left adjoint $L_!$ to the pullback $L^\ast$ in more detail. 
Instead of working on all of $\Loc$, we consider only the 
subcategory
\begin{flalign}\label{eqn:RCEcategory}
\CC \,:=\, \left(
\parbox{3em}{\xymatrix@R=1.2em@C=1.2em{
&\ar[dl]_-{i_+}M_+ \ar[dr]^-{j_+}&\\
M&&M_h\\
&\ar[ul]^-{i_-} M_-\ar[ur]_-{j_-}&
}}
\right) \,\subseteq\,\Loc
\end{flalign}
that is relevant for determining the relative Cauchy evolution induced by a pair $(M,h)$.
Because $\CC$ does not contain any causally disjoint pairs of morphisms,
i.e.\ the restriction of $\perp_{\Loc}^{}$ to $\CC$ is empty, 
the category $\AQFT(\ovr{\CC}) \simeq \Fun(\CC,\Alg_\bbK)$ of AQFTs on $\ovr{\CC}$
is simply the category of all functors from $\CC$ to the category of associative
and unital $\bbK$-algebras $\Alg_\bbK$. Furthermore, because all morphisms in $\CC$ 
are Cauchy morphisms, the analog of the orthogonal localization $\ovr{\Loc}[W^{-1}]$
in the present context is given by the localization $\CC[\mathrm{All}^{-1}]$ 
at all $\CC$-morphisms, which also has an empty orthogonality relation.
\sk

The localization $\CC[\mathrm{All}^{-1}]$ can be described very explicitly. Let 
us denote by $\BB \bbZ$ the category with a single object $\ast$ and morphisms given 
by the group $\bbZ$ of integers. It is useful to note that any functor 
$\BBB: \BB\bbZ\to\DD$ to any category $\DD$ can be described equivalently 
as an object $\BBB(\ast)\in\DD$, which
we shall often denote simply as $\BBB=\BBB(\ast)$, equipped with a $\bbZ$-action 
$\bbZ\ni k\mapsto \BBB(k)\in\Aut(\BBB(\ast))$. Now consider the functor
\begin{flalign}
\nn L\,:\, \CC ~&\longrightarrow~\BB\bbZ\quad,\\
\nn M,M_-,M_h,M_+ ~&\longmapsto~\ast\quad,\\
\nn i_-~&\longmapsto~1\quad,\\
j_-,j_+,i_+~&\longmapsto~ 0\quad. \label{eqn:localizationfunctor}
\end{flalign}
\begin{lem}
The functor \eqref{eqn:localizationfunctor} is a localization of the category 
$\CC$ given in \eqref{eqn:RCEcategory} at all morphisms.
\end{lem}
\begin{proof}
We have to confirm that the functor $L  : \CC\to \BB\bbZ$ satisfies
the properties of a localization functor, see e.g.\ \cite[Section 7.1]{KashiwaraSchapira}.
Because $\BB\bbZ$ is a groupoid, it is clear that 
$L$ sends every $\CC$-morphism to an isomorphism.
Let now $F : \CC\to \DD$ be any functor to a category $\DD$ that sends every $\CC$-morphism 
to an isomorphism in $\DD$. We have to construct a functor $\widetilde{F} : \BB\bbZ\to\DD$
and a natural isomorphism $\widetilde{F} L \cong F$. Let us define
$\widetilde{F} : \BB\bbZ \to \DD$ by $\widetilde{F}(\ast) := F(M)$
and, for $n\in\bbZ$, $\widetilde{F}(n):= \widetilde{F}(1)^n$ with
$\widetilde{F}(1) := F(i_-)\,F(j_-)^{-1}\,F(j_+)\,F(i_+)^{-1} : F(M)\to F(M)$.
Then a natural isomorphism $\eta : F\to \widetilde{F} L$ can be defined by the components
\begin{flalign}
\nn \eta_{M} := \id_{F(M)}\,:\, F(M)~&\longrightarrow~F(M)\quad,\\
\nn \eta_{M_+} := F(i_+)\,:\, F(M_+)~&\longrightarrow~F(M)\quad,\\
\nn \eta_{M_h} := F(i_+)\, F(j_+)^{-1}\,:\, F(M_h)~&\longrightarrow~F(M)\quad,\\
\eta_{M_-} := F(i_+)\,F(j_+)^{-1}\, F(j_-):\, F(M_-)~&\longrightarrow~F(M)\quad.
\end{flalign}
It remains to prove that the pullback functor
\begin{flalign}
\nn L^\ast\,:\, \Fun(\BB\bbZ,\DD)~&\longrightarrow~\Fun(\CC,\DD)\quad,\\
\nn (G : \BB\bbZ\to \DD) ~&\longmapsto~(G L : \CC\to \DD)\quad,\\
(\zeta : G\to H) ~&\longmapsto~(\zeta L : G L\to H L)
\end{flalign}
is fully faithful. Faithfulness is a consequence of $L$ being 
surjective on objects.
Concerning fullness, let $\kappa : GL\to HL$ be any natural transformation.
Naturality with respect to the three morphisms $i_+$, $j_-$ and $j_+$ implies that all components coincide,
i.e.\ $\kappa_{M} = \kappa_{M_-} = \kappa_{M_h} = \kappa_{M_+} : G(\ast) \to H(\ast)$,
and naturality with respect to $i_-$ implies that
$\zeta_\ast := \kappa_M : G(\ast)\to H(\ast)$ defines a natural transformation
$\zeta: G\to H$ between functors from $\BB\bbZ$ to $\DD$ such that
$\zeta L = \kappa$.
\end{proof}

In the present context, the adjunction from Proposition \ref{prop:localizationequivalence}
reduces to the adjunction
\begin{flalign}\label{eqn:LanSec2}
\xymatrix@C=3.5em{
L_! \,:\, \Fun(\CC,\Alg_\bbK)~\ar@<0.75ex>[r]&\ar@<0.75ex>[l]  ~\Fun(\BB\bbZ,\Alg_\bbK) \,:\, L^\ast
}\quad,
\end{flalign}
where the right adjoint $L^\ast$ is the pullback functor along $L: \CC\to\BB\bbZ$
and the left adjoint $L_!$ is the left Kan extension along $L: \CC\to\BB\bbZ$.
This adjunction restricts to an adjoint equivalence between $\Fun(\BB\bbZ,\Alg_\bbK)$ and
the full subcategory $\Fun(\CC,\Alg_\bbK)^{\mathrm{All}}\subseteq \Fun(\CC,\Alg_\bbK)$ 
of all functors that assign to every $\CC$-morphism an isomorphism.
The left Kan extension $L_!(\AAA) : \BB\bbZ\to\Alg_\bbK$ of any functor $\AAA : \CC\to\Alg_\bbK$
can be computed as the colimit
\begin{flalign}\label{eqn:L!AAA}
L_!(\AAA)\,:=\, \colim\Big(
\xymatrix{
L/\ast \ar[r]^-{\pi} & \CC \ar[r]^-{\AAA} & \Alg_\bbK
}\Big)\quad,
\end{flalign}
where $L/\ast$ denotes the comma category of $L : \CC\to\BB\bbZ$
over the single object $\ast\in\BB\bbZ$
and $\pi: L/\ast\to \CC$ is the forgetful functor.
More explicitly, the objects of the category $L/\ast$
are pairs $(n,N)\in\bbZ\times\CC$ consisting of an object $N\in\CC$ 
and a morphism $n : L(N)=\ast\to\ast$ in $\BB\bbZ$, and the
morphisms $(n,N)\to(n^\prime,N^\prime)$ are $\CC$-morphisms
$f:N\to N^\prime$ such that $n=L(f)+n^\prime$. Furthermore, 
the forgetful functor $\pi: L/\ast\to\CC\,,~(n,N)\mapsto N$ 
forgets the integers. Observe that
the functor $\pi: L/\ast\to \CC$ may be visualized as a kind of ``universal cover''
over the category $\CC$ given in \eqref{eqn:RCEcategory},
which takes the form of a spiral of zig-zags lying over $\CC$. 
Straightening this spiral to a line, $L/\ast$ may be displayed as
\begin{flalign}\label{eqn:spiral}
\xymatrix@C=1.5em{
\cdots ~~ (n-1,M) & \ar[l]_-{i_-}(n,M_-) \ar[r]^-{j_-}& (n,M_h) &\ar[l]_-{j_+} (n,M_+) \ar[r]^-{i_+}& (n,M) & \ar[l]_-{i_-}(n+1,M_-)~~ \cdots
}\quad.
\end{flalign}
The $\bbZ$-action on $L_!(\AAA)$ \eqref{eqn:L!AAA} is defined
through the universal property of colimits by
\begin{flalign}
\xymatrix@C=4em{
L_!(\AAA) \ar[r]^-{L_!(\AAA)(k)} &L_!(\AAA)\\
\ar[u]^-{\iota_{(n,N)}}\AAA(N)\ar[ur]_-{~~\iota_{(n+k,N)}}~&~
}
\end{flalign}
for all $k\in\bbZ$ and $(n,N)\in L/\ast$,
where the $\iota$'s denote the canonical morphisms into the colimit.
Note that this may be visualized by moving $k\in\bbZ$ levels 
up in the spiral \eqref{eqn:spiral}. This completes the description of
$L_!$ on objects. If $\zeta:\AAA \to\AAA^\prime$ is a morphism in 
$\Fun(\CC,\Alg_\bbK)$ (i.e., a natural transformation) 
then $L_!(\zeta):L_!(\AAA)\to L_!(\AAA^\prime)$ is a 
natural transformation with the single component $L_!(\zeta)_\ast$
determined by functoriality of the colimit \eqref{eqn:L!AAA}.
Concretely, $L_!(\zeta)_\ast\, \iota_{(n,N)} = \iota^\prime_{(n,N)}\, \zeta_N$,
for all $(n,N)\in L/\ast$, where $\iota^{\prime}_{(n,N)} : \AAA^\prime(N) \to L_!(\AAA^\prime)$
are the canonical morphisms for $\AAA^\prime$.
\sk

Given now any $\AAA \in \Fun(\CC,\Alg_\bbK)^{\mathrm{All}}$
that satisfies the time-slice axiom, we know that
\eqref{eqn:L!AAA} provides an equivalent description of this theory
because restricting \eqref{eqn:LanSec2} to $\Fun(\CC,\Alg_\bbK)^{\mathrm{All}}$
is an adjoint equivalence. From this perspective, the relative Cauchy evolution
is given (applying the analogs of~\eqref{eqn:RCEB} and~\eqref{eqn:rLoclocalized} 
to $\BBB=L_!(\AAA)$ and using the definition~\eqref{eqn:localizationfunctor}) 
as the action of the generator $1\in\bbZ$, i.e.\
\begin{flalign}\label{eqn:RCESec2L!}
L_!(\AAA)(1)\,:\,L_!(\AAA)~\longrightarrow~L_!(\AAA)\quad,
\end{flalign}
which may be visualized by moving $1$ level up in the spiral \eqref{eqn:spiral}.
The concept of relative Cauchy evolution for homotopy AQFTs 
that we develop in this paper is based on a higher categorical 
generalization of this particular perspective.
\sk

To conclude, let us briefly note that \eqref{eqn:RCESec2L!}
can also be related more directly to the ordinary relative Cauchy evolution 
in \eqref{eqn:RCEA}. Take any object of $L/\ast$, for instance $(0,M)\in L/\ast$ at level $0$.
Then the canonical morphism $\iota_{(0,M)}: \AAA(M)\to L_!(\AAA)$
is an isomorphism which intertwines
the two descriptions of the relative Cauchy evolution 
given in \eqref{eqn:RCEA} and \eqref{eqn:RCESec2L!}, i.e.\ the diagram
\begin{flalign}
\xymatrix@C=4em{
L_!(\AAA) \ar[r]^-{L_!(\AAA)(1)} & L_!(\AAA)\\
\ar[u]^-{\iota_{(0,M)}} \AAA(M)\ar[r]_-{\RCE_{(M,h)}} &\AAA(M)\ar[u]_-{\iota_{(0,M)}}
}
\end{flalign}
commutes.


\section{\label{sec:proposalhoRCE}Relative Cauchy evolution for homotopy AQFTs}
The aim of this section is to propose a general strategy
for how the relative Cauchy evolution may be defined for
homotopy AQFTs in the sense of \cite{BSWhomotopy}.
The latter are higher categorical refinements of traditional AQFTs
that are designed to capture the higher categorical structures
of gauge theories. Models constructed via the BRST/BV formalism 
\cite{FredenhagenRejzner,FredenhagenRejzner2}
define examples of homotopy AQFTs on $\ovr{\Loc}$. A very special class
of such examples, the linear quantum gauge theories, admit a particularly 
elegant and rigorous construction as homotopy AQFTs using
suitable techniques from homotopical algebra and derived geometry, see
\cite{LinearYM} and also \cite{BruinsmaSchenkel}.
Let us also refer the reader to
\cite{BSreview} for a less technical introduction to this subject.
\sk

We start by considering the relative Cauchy evolution
induced by a single pair $(M,h)$ which, as we have explained in
the previous section, is controlled by the subcategory $\CC\subseteq \Loc$
displayed in \eqref{eqn:RCEcategory}. Because this particular
$\CC$ has an empty orthogonality relation, the category of homotopy AQFTs 
on $\ovr{\CC}=(\CC,\emptyset)$ admits a simple description
\begin{flalign}\label{eqn:inftyAQFTC}
\AQFT_\infty(\ovr{\CC})\,\simeq\,\Fun(\CC,\dgAlg_\bbK)
\end{flalign}
in terms of functors from $\CC$ to the category $\dgAlg_{\bbK}$
of associative and unital dg-algebras over $\bbK$.
Under this equivalence, the model category structure on $\AQFT_\infty(\ovr{\CC})$ 
from \cite{BSWhomotopy} gets identified with the projective
model structure on the functor category.
More explicitly, a morphism $\zeta : \AAA\to \AAA^\prime$
in $\Fun(\CC,\dgAlg_\bbK)$, i.e.\ a natural transformation,
is a weak equivalence (respectively, a fibration)
iff each component $\zeta_{N}: \AAA(N)\to \AAA^\prime(N)$,
for $N\in\CC$, is a quasi-isomorphism between the underlying chain complexes 
(respectively, a degree-wise surjection).
The cofibrations are determined by the lifting properties in 
model categories, see e.g.\ \cite{Hovey}; they are the morphisms
with the left lifting property against all weak equivalences that are also fibrations.
\sk

As evidenced by the concrete examples constructed in \cite{LinearYM} and 
also in \cite{FredenhagenRejzner,FredenhagenRejzner2},
homotopy AQFTs satisfy a priori only a higher categorical relaxation 
of the time-slice axiom, which was called the {\em homotopy time-slice axiom}.
In the present context \eqref{eqn:inftyAQFTC}, 
this means that the functor $\AAA : \CC\to\dgAlg_\bbK$
assigns to every (necessarily Cauchy) morphism $f: N\to N^\prime$ in the category $\CC$ given
in \eqref{eqn:RCEcategory} 
a {\em weak equivalence} $\AAA(f) : \AAA(N)\to\AAA(N^\prime)$
in the model category $\dgAlg_\bbK$, i.e.\ a dg-algebra morphism
whose underlying chain map is a quasi-isomorphism. 
We denote by $\Fun(\CC,\dgAlg_\bbK)^{\mathrm{hoAll}}\subseteq \Fun(\CC,\dgAlg_\bbK)$
the full subcategory of such functors. (One could also introduce as in \cite{Carmona} 
a Bousfield localized model structure on $\Fun(\CC,\dgAlg_\bbK)$
in order to capture the homotopy time-slice axiom. This
provides in general a more refined structure,  which however 
is not necessary for the purpose of the present paper.) 
Given any homotopy AQFT $\AAA\in\Fun(\CC,\dgAlg_\bbK)^{\mathrm{hoAll}}$
satisfying the homotopy time-slice axiom, the traditional
construction of the RCE automorphism in \eqref{eqn:RCEA}
is obstructed by the fact that weak equivalences
$\AAA(f) : \AAA(N)\to\AAA(N^\prime)$ between dg-algebras 
do not in general admit strict inverses. While they do admit 
quasi-inverses in the form of $A_\infty$-quasi-isomorphisms 
\cite{LodayVallette}, these have the disadvantage of
being difficult to work with. 
\sk

Our strategy is therefore
to adapt the third of the equivalent perspectives on the ordinary relative 
Cauchy evolution from Section \ref{sec:ordinaryRCE} 
to the setting of homotopy AQFTs. The analog of the adjunction \eqref{eqn:LanSec2}
for homotopy AQFTs on $\ovr{\CC}$ is given by the Quillen adjunction
\begin{flalign}\label{eqn:LanSec3}
\xymatrix@C=3.5em{
L_! \,:\, \Fun(\CC,\dgAlg_\bbK)~\ar@<0.75ex>[r]&\ar@<0.75ex>[l]  ~\Fun(\BB\bbZ,\dgAlg_\bbK) \,:\, L^\ast
}\quad,
\end{flalign}
where $L^\ast$ is the pullback functor and $L_!$ the left Kan extension along
the localization functor $L:\CC\to \BB\bbZ$. Suppose for the moment that we could prove that
the restriction of the derived unit to 
$\Fun(\CC,\dgAlg_\bbK)^{\mathrm{hoAll}}\subseteq \Fun(\CC,\dgAlg_\bbK)$ 
is a natural weak equivalence. Then each homotopy AQFT
$\AAA\in \Fun(\CC,\dgAlg_\bbK)^{\mathrm{hoAll}}$
satisfying the homotopy time-slice axiom is weakly equivalent
to the object $L^\ast \bbL L_!(\AAA)\in \Fun(\CC,\dgAlg_\bbK) $, 
where $\bbL L_!$ denotes the left derived functor.
Observe that the theory $L^\ast \bbL L_!(\AAA)$ satisfies the strict time-slice axiom
because it lies in the image of the pullback functor $L^\ast$,
and hence its relative Cauchy evolution is simply given in complete analogy to \eqref{eqn:RCESec2L!}
by the action $\bbL L_!(\AAA)(1) : \bbL L_!(\AAA)\to \bbL L_!(\AAA)$ of the generator $1\in\bbZ$.
In particular, this is a strict $\bbZ$-action in terms of strict dg-algebra automorphisms,
in contrast to the homotopy coherent $\bbZ$-action in terms of $A_\infty$-quasi-automorphisms
that would arise by generalizing directly the traditional construction \eqref{eqn:RCEA}
to homotopy AQFTs. In other words, proving that the derived unit
of the Quillen adjunction \eqref{eqn:LanSec3} restricts
on $\Fun(\CC,\dgAlg_\bbK)^{\mathrm{hoAll}}\subseteq \Fun(\CC,\dgAlg_\bbK)$ 
to a natural weak equivalence would provide a rectification theorem
that allows us to strictify the homotopy time-slice axiom.
\sk

Proving such a result, and thereby obtaining
the simple description of the relative Cauchy evolution for homotopy AQFTs outlined above, 
is a highly technical task that we currently do not know how to complete in full generality.
In this work we present a concrete solution for a simplified version of the problem
that allows us to discuss the relative Cauchy evolution 
in the context of {\em linear} homotopy AQFTs \cite{LinearYM}.
The latter are functors $\AAA : \CC\to \dgAlg_\bbK$ that are given by a composition
\begin{flalign}
\AAA \,:\, \xymatrix@C=3em{
\CC \ar[r]^-{(\LLL,\tau)} ~&~ \PoCh_\bbK \ar[r]^-{\CCR} ~&~ \dgAlg_\bbK
}
\end{flalign}
of the canonical commutation relation (CCR) quantization functor $\CCR$ 
(in the context of chain complexes \cite{LinearYM})
and a functor $(\LLL,\tau)$ that assigns chain complexes of {\em linear observables} together
with their Poisson structure. More precisely, the category $\PoCh_\bbK$ of Poisson chain 
complexes is defined as follows: Objects are pairs $(V,\tau)$ consisting of a chain complex
$V\in\Ch_\bbK$ and a chain map $\tau : V\wedge V\to \bbK$ from the antisymmetrized tensor
product, while morphisms $f : (V,\tau)\to (W,\sigma)$ are chain maps $f: V\to W$ 
such that $\sigma \circ (f\wedge f) = \tau$.
Our strategy is to prove an analog of the desired rectification theorem
for linear observables, which leads to a very explicit and strict model for the relative 
Cauchy evolution for $\LLL$. This will be achieved in Section~\ref{sec:rectification}. Then, 
in Section~\ref{sec:Poisson}, we will endow this strictified model with a suitable
Poisson structure and lift our construction along the $\CCR$-functor to obtain
a relative Cauchy evolution for the linear homotopy AQFT $\AAA = \CCR(\LLL,\tau) $. 
We emphasize that this relative Cauchy evolution can be computed explicitly 
as we will do in Section~\ref{sec:Example} for linear Yang-Mills theories. 
\sk

To conclude this section, we would like to mention briefly that the main
idea behind our approach admits a potential generalization to
describe the whole coherent family $\{\RCE_{(M,h)}^{}\}$
of RCE automorphisms in the context of homotopy AQFT, 
in contrast to only a single RCE automorphism $\RCE_{(M,h)}$.
For this we consider the Quillen adjunction
\begin{flalign}\label{eqn:Ladjunction}
\xymatrix@C=3.5em{
L_! \,:\, \AQFT_\infty(\ovr{\Loc})~\ar@<0.75ex>[r]&\ar@<0.75ex>[l]  ~\AQFT_{\infty}(\ovr{\Loc}[W^{-1}]) \,:\, L^\ast
}
\end{flalign}
introduced in \cite{BSWhomotopy,BruinsmaSchenkel}, which relates
homotopy AQFTs on $\ovr{\Loc}$ to homotopy AQFTs on the orthogonal 
localization $\ovr{\Loc}[W^{-1}]$ at all Cauchy morphisms.
If one could prove that the derived unit of \eqref{eqn:Ladjunction} restricts 
on $\AQFT_\infty(\ovr{\Loc})^{\mathrm{ho}W}\subseteq  \AQFT_\infty(\ovr{\Loc})$
to a natural weak equivalence, then $L^\ast \bbL L_!(\AAA)$ would provide a strictified 
model for the whole coherent family $\{\RCE_{(M,h)}\}$ of RCE automorphisms for a
theory $\AAA\in \AQFT_\infty(\ovr{\Loc})^{\mathrm{ho}W}$.


\section{\label{sec:rectification}Rectification theorem for linear observables}
In this section we shall prove a rectification theorem 
for linear observables. The latter are modeled by functors $\LLL:\CC\to\Ch_\bbK$
from the category $\CC$ \eqref{eqn:RCEcategory} to the model category
of (possibly unbounded) chain complexes of vector spaces over $\bbK$.
In analogy to \eqref{eqn:LanSec3},
the localization functor \eqref{eqn:localizationfunctor} defines a Quillen adjunction
\begin{flalign}\label{eqn:linearadjunction}
\xymatrix@C=3.5em{
L_! \,:\, \Fun(\CC,\Ch_\bbK)~\ar@<0.75ex>[r]&\ar@<0.75ex>[l]  ~\Fun(\BB\bbZ,\Ch_\bbK) \,:\, L^\ast
}\quad,
\end{flalign}
where both sides are endowed with the projective model structure.
Again, $L^\ast$ is simply the pullback functor and $L_!$ the 
left Kan extension along the localization functor $L: \CC\to \BB\bbZ$.
The analog of the homotopy time-slice axiom for functors  $\LLL : \CC\to \Ch_\bbK$ 
assigning chain complexes of linear observables is given by the property that $\LLL$
assigns to every morphism in $\CC$ a weak equivalence (i.e.\ a quasi-isomorphism) in $\Ch_\bbK$. 
We denote by $\Fun(\CC,\Ch_\bbK)^{\mathrm{ho}\mathrm{All}}\subseteq \Fun(\CC,\Ch_\bbK)$
the full subcategory of such functors.
The rectification problem for linear observables on $\CC$ is then given
by the question of whether the derived unit of \eqref{eqn:linearadjunction} restricts 
on $\Fun(\CC,\Ch_\bbK)^{\mathrm{ho}\mathrm{All}}\subseteq \Fun(\CC,\Ch_\bbK)$
to a natural weak equivalence.
\sk

In order to prove our rectification theorem, we require
explicit models for the derived functors associated with the 
Quillen adjunction \eqref{eqn:linearadjunction} and also for their derived unit and counit.
Recall that the left and right derived functors, in the context of a Quillen 
adjunction between model categories, are compositions $\bbL L_!=L_! Q$, $\bbR L^\ast=L^\ast R$ 
where $Q$ and $R$ are cofibrant and fibrant replacement functors on $\Fun(\CC,\Ch_\bbK)$ 
and $\Fun(\BB\bbZ,\Ch_\bbK)$ respectively.
Because each object in $\Fun(\BB\bbZ,\Ch_\bbK)$ is fibrant in the projective model structure,
we may take $R=\id$, so the right derived functor is 
the ordinary right adjoint, $\bbR L^\ast := L^\ast$. Concerning the left derived functor
\begin{flalign}
\bbL L_!\,:\, \Fun(\CC,\Ch_\bbK) ~\longrightarrow~ \Fun(\BB\bbZ,\Ch_\bbK)\quad,
\end{flalign}
we shall use the bar resolution techniques developed by Fresse 
in \cite[Theorem 17.2.7 and Section 13.3]{Fresse}. 
Following the general construction outlined in Appendix \ref{app:bar},
we obtain that the action of the left derived functor on an object 
$X\in\Fun(\CC,\Ch_\bbK)$ is given by
\begin{flalign}\label{eqn:leftderivedfunctor}
\bbL L_!(X)\,:=\, \Tot(\widetilde{X})\,:= \,\Tot\big(\overline{B_{\Delta}}(\BB\bbZ,\CC,X)\big)\,\in\,
\Fun(\BB\bbZ,\Ch_\bbK)\quad,
\end{flalign}
where we abbreviate by $\widetilde{X}:= \overline{B_{\Delta}}(\BB\bbZ,\CC,X)\in\Fun(\BB\bbZ,\bCh_\bbK)$
the bicomplex-valued functor described in \eqref{eqn:bargeneral} and $\Tot$ denotes the
$\bigoplus$-totalization of bicomplexes (see also Appendix \ref{app:bimcomplexes}).
In the present case, the functor $\widetilde{X}$ admits the following very explicit description
as a bicomplex with a $\mathbb{Z}$-action:
The underlying bicomplex
\begin{subequations}\label{eqn:tildeXbicomplex}
\begin{flalign}
\widetilde{X} \,=\, \Big(
\xymatrix{
\widetilde{X}_{0,\bullet} ~&~\ar[l]_-{\delta}\widetilde{X}_{1,\bullet}
}\Big) \,\in\,\bCh_{\bbK}
\end{flalign}
is concentrated only in vertical degrees $0$ and $1$ because,
due to the specific form of the category $\CC$ in \eqref{eqn:RCEcategory},
there exist no composable $m$-tuples $(f_1,\dots,f_m)\in\mathrm{Mor}_m(\CC)$
of $\CC$-morphisms with each $f_i\neq \id$ in the case of $m>1$.
It is concretely given by direct sums of chain complexes
\begin{flalign}\label{eqn:tildeXbicomplex_dirsums}
\widetilde{X}_{0,\bullet} \,=\,\bigoplus_{n\in\bbZ}\bigoplus_{N\in\CC} X(N)_\bullet \quad,\qquad
\widetilde{X}_{1,\bullet}\,=\, \bigoplus_{n\in\bbZ} \bigoplus_{\mycom{f\in \mathrm{Mor}(\CC)}{f\neq \id}} X(\mathsf{s}f)_\bullet \quad,
\end{flalign}
where $\mathsf{s}f, \mathsf{t}f \in\CC$ denote the source/target 
of the morphism $(f: \mathsf{s}f \to \mathsf{t} f) \in\mathrm{Mor}(\CC)$. (Here, 
we have simplified the triple direct sum appearing in~\eqref{eqn:bargeneralmge1} in an obvious way.) 
Note that the direct sum in $\widetilde{X}_{0,\bullet}$ is over the objects of the spiral
\eqref{eqn:spiral} and that the direct sum in $\widetilde{X}_{1,\bullet}$ is over its  non-identity
morphisms. We shall use the notation $(n,N,x)\in \widetilde{X}_{0,\bullet}$
to denote the element $x\in X(N)$ in the summand indexed by $(n,N)$, and similarly
$(n,f,x)\in \widetilde{X}_{1,\bullet}$. The vertical differential $\delta$
is given by
\begin{flalign}
\delta\big(n,f,x\big)\,=\,(-1)^{\vert x\vert} \,\Big( \big(n+L(f) , \mathsf{s}f, x\big) - \big(n,\mathsf{t} f, X(f)x\big)\Big)\quad,
\end{flalign}
\end{subequations}
for all $(n,f,x)\in \widetilde{X}_{1,\bullet}$, 
where $\vert x\vert$ denotes the degree of $x\in X(\mathsf{s}f)$,
and the $\bbZ$-action on the bicomplex $\widetilde{X}$ is given
by addition, i.e.\ $\widetilde{X}(k) :\widetilde{X} \to \widetilde{X}$, 
for $k\in \bbZ$, maps  $(n,N,x)\mapsto(k+n,N,x)$ and $(n,f,x)\mapsto (k+n,f,x)$.
\sk

We still have to provide explicit models for the derived unit and 
counit of the Quillen adjunction \eqref{eqn:linearadjunction}.
The component at $Y\in \Fun(\BB\bbZ,\Ch_\bbK)$ of the derived counit is 
given by (see Appendix~\ref{app:bar}\footnote{To obtain \eqref{eqn:epsilontilde} and 
\eqref{eqn:etatilde} from Appendix~\ref{app:bar} one should keep in mind the way 
that the various multiple direct sums have been simplified to 
obtain~\eqref{eqn:tildeXbicomplex_dirsums} and~\eqref{eqn:Xdelta_dirsums}.})
\begin{flalign}
\nn \epsilon_Y\,:\,\bbL L_! (L^\ast Y)~&\longrightarrow~ Y\quad,\\
\nn (n,N,y)~&\longmapsto~Y(n)y\quad,\\
(n,f,y)~&\longmapsto~0\quad,\label{eqn:epsilontilde}
\end{flalign}
which clearly is $\bbZ$-equivariant, i.e.\ a
natural transformation between functors from $\BB\bbZ$ to $\Ch_\bbK$.
(Note that every complex $L^\ast Y(N)$ is a copy of $Y(\ast)$.)
\sk

The component at $X\in \Fun(\CC,\Ch_\bbK)$
of the derived unit, $\eta_{X} : Q(X)\to L^\ast \bbL L_! (X)$,
is a morphism whose source 
\begin{flalign}
Q(X) \,:=\,  \Tot(X^\Delta)\,:=\, \Tot\big(\overline{B_\Delta}(\CC,\CC,X)\big)\,\in\,\Fun(\CC,\Ch_\bbK) 
\end{flalign}
is a resolution of $X$ that is determined by the bar construction for $\id: \CC\to \CC$. 
Following again the general construction outlined in Appendix \ref{app:bar},
the object $X^\Delta := \overline{B_\Delta}(\CC,\CC,X)  \in \Fun(\CC,\bCh_\bbK)$ is given 
explicitly by, for all $N\in\CC$,
\begin{subequations}\label{eqn:Xdelta}
\begin{flalign}
X^\Delta(N) \,=\, \Big(
\xymatrix{
X^\Delta(N)_{0,\bullet} ~&~\ar[l]_-{\delta} X^\Delta(N)_{1,\bullet}
}\Big) \,\in\,\bCh_{\bbK}\quad,
\end{flalign}
with (simplifying the appropriate double and triple direct sums from~\eqref{eqn:bargeneral})
\begin{flalign}\label{eqn:Xdelta_dirsums}
X^\Delta(N)_{0,\bullet} \,=\, \bigoplus_{\mycom{g\in\mathrm{Mor}(\CC)}{\mathsf{t}g=N}} X(\mathsf{s}g)_\bullet \quad,\qquad
X^\Delta(N)_{1,\bullet} \,=\, \bigoplus_{\mycom{(g,f)\in \mathrm{Mor}_2(\CC)}{\mathsf{t}g=N\,,\,f\neq \id}}  X(\mathsf{s}f)_\bullet\quad,
\end{flalign}
where $ \mathrm{Mor}_2(\CC)$ denotes
the set of composable pairs of morphisms, and vertical differential
\begin{flalign}
\delta\big(g,f,x\big) \,=\, (-1)^{\vert x\vert}\,\Big(\big(gf,x\big) - \big(g, X(f)x\big)\Big)\quad,
\end{flalign}
\end{subequations}
for all $(g,f,x)\in X^\Delta(N)_{1,\bullet}$. 
(Note that since the nerve of $\CC$ is degenerate in degrees $\geq 2$, it follows that
each element $(g,f,x)\in X^\Delta(N)_{1,\bullet}$ is necessarily of the form $(g,f,x)=(\id_N,f,x)$.)
The functor structure on $X^\Delta$
is given by post-composition, i.e.\ for every morphism $h: N\to N^\prime$ in $\CC$, 
the $\bCh_\bbK$-morphism $X^\Delta(h): X^\Delta(N)\to X^\Delta(N^\prime)$
maps $(g,x)\mapsto (hg,x)$ and $(g,f,x)\mapsto (hg,f,x)$.
The component at $X\in \Fun(\CC,\Ch_\bbK)$ of the derived unit is then given by
the components
\begin{flalign} 
\nn \eta_{X,N} \,:\, Q(X)(N) ~&\longrightarrow~\bbL L_!(X)\quad,\\
\nn (g,x)~&\longmapsto ~ \big(L(g),\mathsf{s}g ,x\big)\quad,\\
(g,f,x)~&\longmapsto~\big(L(g),f,x\big)\quad, \label{eqn:etatilde}
\end{flalign}
for all $N\in\CC$, which determine a natural transformation between functors
from $\CC$ to $\Ch_\bbK$.
\begin{rem}\label{rem:qmap}
We would like to emphasize that $Q(X) = \Tot(X^\Delta)$ is indeed naturally weakly 
equivalent to $X\in \Fun(\CC,\Ch_\bbK)$ via the natural weak equivalence $q_X : Q(X)\to X$
determined by the components
\begin{flalign}
\nn q_{X,N}\,:\, Q(X)(N) ~&\longrightarrow~X(N)\quad,\\
\nn (g,x)~&\longmapsto ~ X(g)x\quad,\\
(g,f,x)~&\longmapsto ~ 0\quad, \label{eqn:qmap}
\end{flalign}
for all $N\in\CC$. See the standard argument in \cite[Lemma 13.3.3]{Fresse}. 
An explicit quasi-inverse
for the component $q_{X,N}: Q(X)(N)\to X(N)$ is given by the chain map
\begin{flalign}\label{eqn:qmapquasiinverse}
s_{X,N} \,:\, X(N)~\longrightarrow~Q(X)(N)~,~~x~\longmapsto~(\id_N, x)\quad,
\end{flalign}
which is of course a quasi-isomorphism too.
\end{rem}

With these preparations, we can now prove the main result of this section.
\begin{theo}\label{theo:rectification}
The derived counit of the Quillen adjunction \eqref{eqn:linearadjunction} is a natural weak equivalence.
The restriction of the derived unit to the full subcategory
$\Fun(\CC,\Ch_\bbK)^{\mathrm{ho}\mathrm{All}}\subseteq \Fun(\CC,\Ch_\bbK)$ of functors
that send every $\CC$-morphism to a quasi-isomorphism
is a natural weak equivalence too.
\end{theo}
\begin{proof}
We first prove that, for every $Y\in\Fun(\BB\bbZ,\Ch_\bbK)$, 
the component $\epsilon_Y : \bbL L_! (L^\ast Y)\to Y$ of the 
derived counit \eqref{eqn:epsilontilde} is a weak equivalence. Our strategy is
to construct an explicit quasi-inverse of its underlying chain map.
We define the chain map 
\begin{flalign}
\kappa \,:\, Y~ \longrightarrow ~\bbL L_! (L^\ast Y)~,~~y~\longmapsto~(0,M,y)\quad,
\end{flalign}
where $M$ denotes the left object in the category $\CC$ in 
\eqref{eqn:RCEcategory} (i.e., the unperturbed spacetime),
which satisfies $\epsilon_Y\,\kappa=\id$. The other composition $\kappa\, \epsilon_Y$
is homotopic to the identity, i.e.\
\begin{flalign}\label{eqn:homotopytmpcounit}
\kappa\, \epsilon_Y-\id = \partial \rho\quad,
\end{flalign}
via a chain homotopy $\rho\in\hom\big(\bbL L_! (L^\ast Y),\bbL L_! (L^\ast Y)\big)_1$
that we shall now describe. Given any element
$(n,N,y)\in \bbL L_! (L^\ast Y)=\Tot(\widetilde{L^\ast Y})$ of vertical degree $0$, 
let us denote the shortest zig-zag in \eqref{eqn:spiral} from $(0,M)$ to $(n,N)$ by
\begin{subequations}\label{eqn:homotopytmpcounit2}
\begin{flalign}
\resizebox{.9 \textwidth}{!} 
{$
\xymatrix@C=1.5em{
(0,M) = (n_{-1},N_{-1}) &\ar[l]_-{f_0} (n_0,N_0) \ar[r]^-{f_1}& (n_1,N_1) &\ar[l]_-{f_2}   (n_2,N_2) \ar[r]^-{f_3} &\cdots &\ar@{<->}[l]_-{f_m} (n_m,N_m)= (n,N)
}
$}
\end{flalign}
and write 
\begin{flalign}
f_i \,:\, (n_i^{\mathsf{s}} , N_i^{\mathsf{s}})~\longrightarrow~ (n_i^{\mathsf{t}},N_i^{\mathsf{t}})
\end{flalign}
\end{subequations} 
for the $i$-th morphism.
We define the chain homotopy $\rho$ by
\begin{flalign}
\rho(n,N,y)\,=\, -(-1)^{\vert y\vert}\sum_{i=0}^m (-1)^i\,\big(n_i^{\mathsf{t}}, f_i, Y(n-n_i^{\mathsf{s}})y\big)\quad,\qquad
\rho(n,f,y)\,=\,0\quad,
\end{flalign}
for all $(n,N,y)\in \bbL L_! (L^\ast Y)$ of vertical degree $0$
and $(n,f,y)\in \bbL L_! (L^\ast Y)$ of vertical degree $1$.
It is straightforward to check that \eqref{eqn:homotopytmpcounit} holds true.
For elements $(n,N,y)\in \bbL L_! (L^\ast Y)$ of vertical degree $0$, we find that
\begin{flalign}
\nn \partial \rho (n,N,y) &= (\delta + \dd)\rho (n,N,y) + \rho (\delta+\dd)(n,N,y)
= \delta \rho (n,N,y)\\
\nn &= - \sum_{i=0}^m(-1)^i \, \Big(\big(n_i^{\mathsf{s}},N_i^{\mathsf{s}},Y(n-n_{i}^\mathsf{s}) y\big)
- \big(n_i^\mathsf{t}, N_i^{\mathsf{t}}, Y(n-n_i^{\mathsf{t}}) y\big) \Big)\\
&= \big(0,M,Y(n)y\big) - \big(n,N,y\big)\quad,
\end{flalign}
where the second equality holds because $Y(k)$ is a chain map, 
and the last step follows from the observation that the terms associated with the inner nodes in
\eqref{eqn:homotopytmpcounit2} cancel out, leaving only the two 
boundary terms associated with the ends of \eqref{eqn:homotopytmpcounit2}.
By a similar calculation, one checks that
\begin{flalign}
 \partial \rho (n,f,y)  = \rho\delta(n,f,y) = -(n,f,y)\quad,
\end{flalign}
for all elements $(n,f,y)\in \bbL L_! (L^\ast Y)$ of vertical degree $1$. This completes the proof
that all components $\epsilon_Y$ of the derived counit are weak equivalences.
\sk

Let us now prove that, for all $X\in \Fun(\CC,\Ch_\bbK)^{\mathrm{ho}\mathrm{All}}$,
the component $\eta_{X} : Q(X)\to L^\ast \bbL L_!(X)$ of the derived unit \eqref{eqn:etatilde}
is a weak equivalence in $\Fun(\CC,\Ch_\bbK)$. This amounts to proving that the underlying 
component chain maps $\eta_{X,N} : Q(X)(N) \to \bbL L_!(X)$
are quasi-isomorphisms, for all objects $N\in\CC$. Because
the category $\CC$ in \eqref{eqn:RCEcategory} is connected and
$Q(X)(f)$ is a quasi-isomorphism for each $\CC$-morphism $f$,
naturality of $\eta_X$ and the $2$-out-of-$3$ property
of weak equivalences implies that it suffices to prove
that one of the components of $\eta_X$ is a quasi-isomorphism.
We consider the left object $M$ in the category $\CC$ in \eqref{eqn:RCEcategory}
and shall prove the equivalent condition that the chain map
\begin{flalign}\label{eqn:phimapping}
\phi :=\, \eta_{X,M}\,s_{X,M} \,:\, X(M)~\longrightarrow~ \bbL L_!(X)~,~~x~\longmapsto ~(0,M,x)
\end{flalign}
obtained by pre-composing with the quasi-isomorphism \eqref{eqn:qmapquasiinverse}
is a quasi-isomorphism. Because $\phi$ is injective, we obtain a short exact sequence
\begin{flalign}\label{eqn:SES}
\xymatrix@C=1.5em{
0 \ar[r]~&~X(M)\ar[r]^-{\phi} ~&~ \bbL L_!(X)\ar[r] ~&~ \mathrm{coker}(\phi)\ar[r] ~&~ 0
}
\end{flalign}
in $\Ch_\bbK$. Our strategy is to prove, via a spectral sequence argument, 
that the homology of $\mathrm{coker}(\phi)$ is trivial, which implies via the long exact 
homology sequence associated to \eqref{eqn:SES}
that $\phi$ is a quasi-isomorphism. The cokernel of \eqref{eqn:phimapping} 
admits a direct sum decomposition
$\mathrm{coker}(\phi) = F^{\mathrm{L}} \oplus F^{\mathrm{R}}$,
where $F^{\mathrm{L}}\in\Ch_\bbK$ is the chain complex associated with 
the objects and morphisms in \eqref{eqn:spiral} that are to the left of $(0,M)$
and $F^{\mathrm{R}}\in\Ch_\bbK$ is associated with the objects and morphisms  in \eqref{eqn:spiral} 
that are to the right of $(0,M)$.  Explicitly, 
using the following convenient notation
\begin{flalign}
\xymatrix@C=1.5em{
\Big( (M,0)& \ar[l]_-{f_0^R} N^R_0
 \ar[r]^-{f_1^R} 
 & N^R_1
 & \ar[l]_-{f_2^R}
  \cdots \Big) 
 : = 
\Big( (M,0)&\ar[l]_-{i_-}
(1, M_- )  
 \ar[r]^-{j_-} 
 & (1, M_h )
 & \ar[l]_-{j_+}
 \cdots \Big) 
}
\end{flalign}
for the objects and morphisms to the right of $(0,M)$,  we can write
\begin{subequations}\label{eqn:FRcomplex}
\begin{flalign}
F^{\mathrm{R}}\,=\,\Tot\Big(\xymatrix{
\bigoplus\limits_{k=0}^\infty X(N^{\mathrm{R}}_k) ~&~\ar[l]_-{\delta^{\mathrm{R}}} \bigoplus\limits_{k=0}^\infty X(\mathsf{s}f^{\mathrm{R}}_k)
}\Big)\quad.
\end{flalign}
The vertical differential $\delta^{\mathrm{R}}$
on $(k,x) \in \bigoplus_{k=0}^\infty X(\mathsf{s}f^{\mathrm{R}}_k)$ is given by
\begin{flalign}
\delta^{\mathrm{R}}(k,x) \,=\,\begin{cases} 
(-1)^{\vert x\vert}\, (0,x) &~,~~\text{for $k=0$}\quad, \\
(-1)^{\vert x\vert}\,\big(\big(k,x\big) - \big(k-1, X(f_{k}^\mathrm{R}) x\big) \big) &~,~~\text{for $k\geq 2$ even}\quad,\\
(-1)^{\vert x\vert}\, \big(\big(k-1,x\big) - \big(k, X(f_k^{\mathrm{R}})x\big)\big) &~,~~\text{for $k$ odd}\quad,
\end{cases}
\end{flalign}
\end{subequations}
where the difference to \eqref{eqn:tildeXbicomplex} for $k=0$ is due to the fact that the object associated with 
$(0,M)$ is divided out in the cokernel.
The chain complex $F^{\mathrm{L}}$ is defined analogously by considering 
objects and morphisms that are to the left of $(0,M)$.
To prove that the homology of $F^{\mathrm{R}}$ is trivial, we consider 
the (bounded below and exhaustive) filtration
$0\subseteq F^{\mathrm{R}}_0 \subseteq F^{\mathrm{R}}_1\subseteq \cdots\subseteq F^{\mathrm{R}}_p\subseteq \cdots 
\subseteq F^{\mathrm{R}}$
given by restricting the direct sums in \eqref{eqn:FRcomplex} to $k\leq p$.
Observe that the quotients
\begin{flalign}
F^{\mathrm{R}}_p\big/F^{\mathrm{R}}_{p-1} \, =\,\begin{cases}
\mathrm{cone}\left(\id : X(N_p^{\mathrm{R}})\to X(N_p^{\mathrm{R}}) \right) &~,~~\text{for $p\geq 0$ even}\quad,\\
\mathrm{cone}\left(-X(f_p^\mathrm{R}) : X(N_{p-1}^{\mathrm{R}})\to X(N_p^{\mathrm{R}})\right) &~,~~\text{for $p\geq 0$ odd}\quad,
\end{cases} 
\end{flalign}
are mapping cone complexes associated with quasi-isomorphisms. (It is in this step where
we use that $X\in \Fun(\CC,\Ch_\bbK)^{\mathrm{ho}\mathrm{All}}$ satisfies the homotopy time-slice axiom.)
Using \cite[Corollary 1.5.4]{Weibel}, we obtain that 
$H_\bullet \big(F^{\mathrm{R}}_p/F^{\mathrm{R}}_{p-1}\big) =0$,
for all $p$, and the convergence theorem for spectral sequences in \cite[Theorem 5.5.1]{Weibel}
then implies that $H_\bullet(F^{\mathrm{R}})=0$. Repeating the same argument for $F^{\mathrm{L}}$
we obtain $H_\bullet(F^{\mathrm{L}})=0$, from which we conclude that $H_\bullet(\mathrm{coker}(\phi))=0$. 
This completes the proof.
\end{proof}

It is worthwhile to recall and explain in more detail the precise sense in which
the results of this section provide a strictification of the homotopy time-slice axiom
for linear observables and how they are useful for defining a concept of relative Cauchy evolution.
Suppose that $\LLL : \CC\to \Ch_\bbK$ is a functor that assigns the
chain complexes of linear observables of a linear homotopy AQFT. (Recall that the Poisson
structure $\tau : \LLL\wedge \LLL\to \bbK$ is neglected in the present 
section and will be considered later in Section \ref{sec:Poisson}.)
By definition, such $\LLL$ satisfies the homotopy time-slice axiom if and only
if it defines an object $\LLL \in \Fun(\CC,\Ch_\bbK)^{\mathrm{ho}\mathrm{All}}$.
Combining Remark \ref{rem:qmap} and Theorem \ref{theo:rectification},
we obtain a zig-zag of weak equivalences
\begin{flalign}\label{eqn:zigzagWE}
\xymatrix@C=4em{
\LLL ~&~ \ar[l]^-{\sim}_-{q_{\LLL}}Q(\LLL) \ar[r]_-{\sim}^-{\eta_\LLL}~&~ L^\ast \bbL L_! (\LLL) 
}
\end{flalign}
in the model category $\Fun(\CC,\Ch_\bbK)$. 
With this zig-zag we can replace our original $\LLL$, which satisfies
the homotopy time-slice axiom, by the equivalent object $L^\ast \bbL L_! (\LLL) $,
which satisfies the {\em strict time-slice axiom} because it arises via the pullback functor $L^\ast$.
From this equivalent perspective, the relative Cauchy evolution takes an extremely simple form
because it is implemented by the $\bbZ$-action of the generator $1\in\bbZ$. 
Explicitly, we obtain the RCE automorphism for linear observables
\begin{flalign}
\nn \rce_{(M,h)}^{}\,:=\, \bbL L_!(\LLL)(1) \,:\, \bbL L_!(\LLL)~&\longrightarrow~\bbL L_!(\LLL)\quad,\\
\nn (n,N,x)~&\longmapsto~(n+1,N,x)\quad,\\
(n,f,x)~&\longmapsto~(n+1,f,x)\quad,\label{eqn:smallRCEmap}
\end{flalign}
which we denote by small letters to distinguish it from the induced
RCE automorphism on CCR-algebras that we shall study in the next section.
Observe that, similarly to the reformulation in \eqref{eqn:RCESec2L!} of the traditional RCE, 
this RCE automorphism is given simply by adding $+1$
to the level $n\in\bbZ$ in the spiral \eqref{eqn:spiral} 
which enters our model for the left derived functor in
\eqref{eqn:leftderivedfunctor} and \eqref{eqn:tildeXbicomplex}.


\section{\label{sec:Poisson}Poisson structure and lift along the $\CCR$-functor}
Our constructions and results in Section \ref{sec:rectification} neglected the Poisson 
structure $\tau : \LLL\wedge \LLL\to \bbK$ on the functor $\LLL : \CC\to \Ch_\bbK$
that describes the assignment of the chain complexes of linear observables of a linear homotopy AQFT.
The aim of this section is to introduce a suitable Poisson structure
$\tau_\bbL : \bbL L_!(\LLL)\wedge \bbL L_!(\LLL)\to \bbK$ on our model
$\bbL L_!(\LLL) = \Tot(\widetilde{\LLL}): \BB\bbZ\to\Ch_\bbK$ for the left derived functor
given in \eqref{eqn:leftderivedfunctor} and \eqref{eqn:tildeXbicomplex}
that is compatible with the original Poisson structure in the following sense:
Using pullbacks along the zig-zag of weak equivalences in \eqref{eqn:zigzagWE},
we can induce two Poisson structures on $Q(\LLL):\CC\to\Ch_\bbK$, namely
\begin{subequations}\label{eqn:PBPoissonhomtopy}
\begin{flalign}
q_\LLL^\ast(\tau):= \tau\circ (q_\LLL\wedge q_\LLL)\,:\, Q(\LLL)\wedge Q(\LLL)~&\longrightarrow \bbK\quad,\\
\qquad \eta_\LLL^\ast (L^\ast\tau_\bbL):= (L^\ast \tau_\bbL)\circ (\eta_\LLL\wedge \eta_\LLL)\,:\,  Q(\LLL)\wedge Q(\LLL)~&\longrightarrow \bbK\quad,
\end{flalign}
where $\tau$ denotes the original Poisson structure on $\LLL:\CC\to\Ch_\bbK$ 
and $L^\ast\tau_\bbL$ the Poisson structure
induced by $\tau_\bbL$ on the pullback $L^\ast \bbL L_!(\LLL):= \bbL L_!(\LLL) \circ L :\CC\to \Ch_\bbK$ .
Our compatibility condition on $\tau_\bbL$
demands the existence of a natural chain homotopy 
$\rho\in \hom(Q(\LLL)\wedge Q(\LLL),\bbK)_1$ such that 
\begin{flalign}
\eta_\LLL^\ast (L^\ast \tau_\bbL) - q_\LLL^\ast(\tau) = \partial\rho\quad,
\end{flalign}
\end{subequations}
where $\partial$ denotes the differential on mapping complexes.
(Here, we recall that if $\kappa\in \hom(V,W)_n$, representing a collection 
of maps $\kappa_m:V_m\to W_{m+n}$ for all $m\in\bbZ$, 
then $\partial\kappa\in\hom(V,W)_{n-1}$ is given by 
$(\partial\kappa)_m =\dd^W\circ\kappa_m - (-1)^n\,\kappa_{m-1}\circ \dd^V$.)
Using  previous results on homotopical properties of the 
$\CCR$-functor from \cite{BruinsmaSchenkel,LinearYM}, it then follows that there exists
a zig-zag of weak equivalences between $\CCR(\LLL,\tau)$
and $\CCR(L^\ast \bbL L_!(\LLL),L^\ast \tau_\bbL)$,
which provides an explicit strictification construction
of the homotopy time-slice axiom of linear homotopy AQFTs
and hence allows us to study their relative Cauchy evolution.
\sk

In order to introduce a Poisson structure $\tau_\bbL$ 
on $\bbL L_!(\LLL)$, we make use of the pictorial interpretation 
from \eqref{eqn:spiral}. Let us also choose, for each 
quasi-isomorphism $\LLL(f) : \LLL(N)\to\LLL(N^\prime)$,
a quasi-inverse chain map $\LLL(f)^{-1}: \LLL(N^\prime)\to \LLL(N)$,
two chain homotopies $\lambda_f\in \hom(\LLL(N^\prime),\LLL(N^\prime))_1$ 
and $\gamma_f\in\hom(\LLL(N),\LLL(N))_1 $, and a chain $2$-homotopy
$\xi_f \in \hom(\LLL(N),\LLL(N^\prime))_2$, such that
\begin{subequations}\label{eqn:homotopy4quasiinverse}
\begin{flalign}
\LLL(f)\,\LLL(f)^{-1} -\id &= \partial \lambda_{f}\quad,\\
\LLL(f)^{-1}\,\LLL(f) - \id &= \partial \gamma_f\quad,\\
\LLL(f)\,\gamma_f - \lambda_f\,\LLL(f) &= \partial \xi_f\quad.\label{eqn:homotopy4quasiinverse3}
\end{flalign}
\end{subequations}
For the case of $f=\id_N$ being an identity morphism, we choose
$\LLL(\id_N)^{-1}=\id_{\LLL(N)}$ the identity, 
together with $\lambda_{\id_N}=0$, $\gamma_{\id_N}=0$ and $\xi_{\id_N} =0$.
\begin{rem}
The homotopy coherence data in \eqref{eqn:homotopy4quasiinverse} always exists by the following argument:
The category $\Ch_\bbK$ may be enriched 
to a strict $2$-category whose $2$-morphisms are homotopy classes of chain homotopies, i.e.\
equivalence classes $[\theta]$ of elements $\theta \in \hom(V,W)_1$ modulo elements
of the form $\partial \ell$ for all $\ell\in\hom(V,W)_2$. Because $\bbK$ has characteristic $0$, 
quasi-isomorphisms of chain complexes are precisely the equivalences in this $2$-category.
Using the standard result in $2$-category theory that each equivalence  in a $2$-category 
can be improved to an adjoint equivalence, see e.g.\ \cite[Exercise 2.2]{Lack}, we obtain for each
quasi-isomorphism $\LLL(f) : \LLL(N)\to\LLL(N^\prime)$
the data in \eqref{eqn:homotopy4quasiinverse}, where
$\LLL(f)^{-1}$ plays the role of the right adjoint, $\gamma_f$ the role 
of the adjunction unit and $-\lambda_f$ the role of the adjunction counit.
The last identity \eqref{eqn:homotopy4quasiinverse3} states one of 
the two triangle identities for this adjunction, while the other
one reads as
\begin{flalign}\label{eqn:homotopy4quasiinverse4}
\gamma_f\, \LLL(f)^{-1}- \LLL(f)^{-1} \lambda_f = \partial \widetilde{\xi}_f\quad,
\end{flalign}
for some chain $2$-homotopy
$\widetilde{\xi}_f \in \hom(\LLL(N^\prime),\LLL(N))_2$. We did not include
this additional coherence condition \eqref{eqn:homotopy4quasiinverse4}
in the main text because it is not needed for our construction.
We already would like to emphasize at this point that for concrete
examples, such as the linear Yang-Mills model \cite{LinearYM}, 
the homotopy coherence data in \eqref{eqn:homotopy4quasiinverse}
may be constructed from Green operators and a partition of unity.
See Section \ref{sec:Example} for the details. In general, the $2$-categorical argument
automates the explicit but tedious linear algebra that 
constructs quasi-inverses and homotopies, based on choices of 
linear projections onto $\ker \dd_i$ within $X_i$ and onto $H_i(X)$ within $\ker \dd_i$.
\end{rem}

Now consider the diagram formed by applying $\LLL$ to the spiral in \eqref{eqn:spiral}. 
Given any two $(n,N),(n^\prime,N^\prime)\in \bbZ\times\CC$, we note that 
there exists a unique chain of zig-zags 
of minimal length between them. Using the chosen quasi-inverses $\LLL(f)^{-1}$ to reverse arrows as 
needed, we obtain {\em zig-zagging chain maps}
\begin{flalign}\label{eqn:zigzaggingmaps}
Z_{(n^\prime,N^\prime)}^{(n,N)} \,:\, \LLL(N^\prime)~\longrightarrow~\LLL(N) 
\end{flalign}
through the shortest zig-zag from $(n^\prime,N^\prime)$ to $(n,N)$.
For example, the zig-zagging chain map from $(n,M_-)$ to $(n,M)$ is 
given by
\begin{subequations}
\begin{flalign}
Z_{(n,M_-)}^{(n,M)}\,=\,\LLL(i_+)\,\LLL(j_+)^{-1}\,\LLL(j_-)\,:\, \LLL(M_-)~\longrightarrow~\LLL(M)
\end{flalign}
and the one from $(n+1,M_-)$ to $(n,M_h)$ is given by
\begin{flalign}
Z_{(n+1,M_-)}^{(n,M_h)}\,=\, \LLL(j_+)\, \LLL(i_+)^{-1}\,\LLL(i_-)\,:\, \LLL(M_-)~\longrightarrow~\LLL(M_h)\quad.
\end{flalign}
\end{subequations}
By definition, the zig-zagging chain maps \eqref{eqn:zigzaggingmaps} are invariant
under the $\bbZ$-action, i.e.\
\begin{flalign}\label{eqn:translationinvariancezigzagging}
Z^{(n+k,N)}_{(n^\prime+k,N^\prime)} = Z^{(n,N)}_{(n^\prime,N^\prime)}\quad,
\end{flalign}
for all $k\in\bbZ$ and $(n,N),(n^\prime,N^\prime)\in\bbZ\times\CC$.
\sk

In Appendix \ref{app:zigzag}, we 
construct, for every $(n,f)\in \bbZ\times\mathrm{Mor}(\CC)$ and $(n^\prime,N^\prime)\in\bbZ\times\CC$, 
from the chain homotopy $\lambda_f$ a {\em zig-zagging homotopy} (for left composition)
\begin{subequations}\label{eqn:zigzagginghomotopyh}
\begin{flalign}
\Lambda^{(n,f)}_{(n^\prime,N^\prime)} \,\in\, \hom\big(\LLL(N^\prime) , \LLL(\mathsf{t}f)\big)_1\quad
\end{flalign}
such that
\begin{flalign}
\LLL(f)\, Z^{(n+L(f),\mathsf{s}f)}_{(n^\prime,N^\prime)} - Z^{(n,\mathsf{t}f)}_{(n^\prime,N^\prime)} \,=\,\partial \Lambda^{(n,f)}_{(n^\prime,N^\prime)}\quad.
\end{flalign}
\end{subequations}
Similarly, we construct, for every $(n,N)\in \bbZ\times\CC$ 
and $(n^\prime,f^\prime)\in\bbZ\times\mathrm{Mor}(\CC)$,
from the chain homotopy $\gamma_f$ a {\em zig-zagging homotopy} (for right composition)
\begin{subequations}\label{eqn:zigzagginghomotopyk}
\begin{flalign}
\Gamma^{(n,N)}_{(n^\prime,f^\prime)} \,\in\, \hom\big(\LLL(\mathsf{s}f^\prime) , \LLL(N)\big)_1
\end{flalign}
such that
\begin{flalign}
Z^{(n,N)}_{(n^\prime,\mathsf{t} f^\prime)}\, \LLL(f^\prime) - Z^{(n,N)}_{(n^\prime+L(f^\prime),\mathsf{s} f^\prime)} 
\,=\,\partial \Gamma^{(n,N)}_{(n^\prime,f^\prime)}\quad.
\end{flalign}
\end{subequations}
Finally, using also the chain $2$-homotopies $\xi_f$, 
the construction in Appendix \ref{app:zigzag} defines,
for every $(n,f),(n^\prime,f^\prime)\in \bbZ\times\mathrm{Mor}(\CC)$,
a {\em zig-zagging $2$-homotopy}
\begin{subequations}\label{eqn:zigzagging2homotopy}
\begin{flalign}
\Xi^{(n,f)}_{(n^\prime,f^\prime)} \,\in\, \hom\big(\LLL(\mathsf{s} f^\prime) , \LLL(\mathsf{t}f)\big)_2
\end{flalign}
such that
\begin{flalign}
\LLL(f)\, \Gamma_{(n^\prime,f^\prime)}^{(n+L(f),\mathsf{s}f)} - \Gamma_{(n^\prime,f^\prime)}^{(n,\mathsf{t}f)} +
\Lambda_{(n^\prime+L(f^\prime),\mathsf{s}f^\prime)}^{(n,f)} - \Lambda^{(n,f)}_{(n^\prime,\mathsf{t}f^\prime)}\,\LLL(f^\prime)\,=\, \partial \Xi^{(n,f)}_{(n^\prime,f^\prime)}\quad.
\end{flalign}
\end{subequations}
The role of this homotopy coherence data is to relate the shortest zig-zags in \eqref{eqn:spiral}
to longer zig-zags that are obtained by additional compositions with morphisms and their quasi-inverses.
Note that this is just part of a higher tower of homotopy coherence data for relating
arbitrary zig-zags to the shortest ones, which however will not be required in our work.
\sk

With these preparations, we can now define a Poisson structure
$\tau_{\bbL}:\bbL L_!(\LLL)\wedge \bbL L_!(\LLL)\to\bbK$ on $\bbL L_!(\LLL)=\Tot(\widetilde{\LLL})$.
Recall that $\tau_{\bbL}$ should be a chain map, therefore obeying
$\tau_{\bbL}\circ \dd^{\Tot(\widetilde{\LLL}\wedge \widetilde{\LLL})}=\dd^{\mathbb{K}}\circ \tau_{\bbL}=0$.
Using a direct sum decomposition into the vertical degrees of
$\widetilde{\LLL}$, such a Poisson structure is given by a family of $n$-chains
\begin{subequations}\label{eqn:bicomplexPoisson}
\begin{flalign}
\Big\{\tau_{\bbL,n}\in \hom\big((\widetilde{\LLL}\wedge\widetilde{\LLL})_{n,\bullet},\bbK\big)_n\Big\}_{n\in\bbZ}
\end{flalign}
that are invariant under the $\bbZ$-action on $\widetilde{\LLL}$ and
(as $\dd^{\Tot(\widetilde{\LLL}\wedge \widetilde{\LLL})}= 
\delta^{\widetilde{\LLL}\wedge\widetilde{\LLL}} + \dd^{\widetilde{\LLL}\wedge \widetilde{\LLL}}$)
satisfy the conditions
\begin{flalign}
\tau_{\bbL,n}\circ\delta^{\widetilde{\LLL}\wedge\widetilde{\LLL}} + \tau_{\bbL,n+1}\circ \dd^{\widetilde{\LLL}\wedge\widetilde{\LLL}}\,=\,0\quad,
\end{flalign}
\end{subequations}
for all $n\in\bbZ$. Because $\widetilde{\LLL}\in \bCh_{\bbK} $ is concentrated 
in vertical degrees $0$ and $1$ (see \eqref{eqn:tildeXbicomplex}) 
and hence $\widetilde{\LLL}\wedge \widetilde{\LLL}$ is concentrated in degrees $0$, $1$ and $2$,
it follows that
\begin{subequations}\label{eqn:taubbL}
\begin{flalign}
\tau_{\bbL,n}=0~~,\quad\text{for all }n\not\in\{0,1,2\}\quad.
\end{flalign}
We define the degree $0$ component by
\begin{flalign}
\tau_{\bbL,0}\Big((n,N,x)\otimes (n^\prime,N^\prime,x^\prime)  \Big) = \tfrac{1}{2}\Big(\tau_{N}\Big(x\otimes Z_{(n^\prime,N^\prime)}^{(n,N)} x^\prime \Big) + \tau_{N^\prime}\Big(Z_{(n,N)}^{(n^\prime,N^\prime)}x \otimes x^\prime\Big)\Big)\quad,
\end{flalign}
for all $(n,N,x),(n^\prime,N^\prime,x^\prime)\in\widetilde{\LLL}_{0,\bullet}$,
where $\tau_N$ and $\tau_{N^{\prime}}$ denote the components of 
the given Poisson structure $\tau: \LLL\otimes \LLL\to \bbK$. The degree
$1$ component is defined by
\begin{flalign}
\resizebox{.9 \textwidth}{!} 
{$
\tau_{\bbL,1}\Big( (n,f,x)\otimes (n^\prime, N^\prime,x^\prime)\Big)=
\tfrac{1}{2}\Big(\tau_{\mathsf{t} f} \Big(\LLL(f) x \otimes \Lambda^{(n,f)}_{(n^\prime,N^\prime)}x^\prime\Big)
+ (-1)^{\vert x^\prime\vert} \,\tau_{N^\prime}\Big( \Gamma^{(n^\prime,N^\prime)}_{(n,f)} x\otimes x^\prime\Big)\Big)~~,
$}
\end{flalign}
for all $(n,f,x)\in \widetilde{\LLL}_{1,\bullet}$ and 
$(n^\prime,N^\prime,x^\prime)\in\widetilde{\LLL}_{0,\bullet}$,
and
\begin{flalign}
\resizebox{.9 \textwidth}{!} 
{$
\tau_{\bbL,1}\Big( (n,N,x)\otimes (n^\prime, f^\prime,x^\prime)\Big) = \tfrac{(-1)^{\vert x\vert}}{2} \Big(
\tau_N\Big(x\otimes \Gamma^{(n,N)}_{(n^\prime,f^\prime)}x^\prime\Big)
+(-1)^{\vert x^\prime\vert} \,
\tau_{\mathsf{t}f^\prime}\Big(\Lambda^{(n^\prime,f^\prime)}_{(n,N)}x\otimes \LLL(f^\prime)x^\prime\Big)
\Big)~~,
$}
\end{flalign}
for all $(n,N,x)\in\widetilde{\LLL}_{0,\bullet} $ and 
$(n^\prime,f^\prime,x^\prime)\in\widetilde{\LLL}_{1,\bullet}$.
Finally, the degree $2$ component is defined by
\begin{flalign}
\resizebox{.9 \textwidth}{!} 
{$
\tau_{\bbL,2}\Big((n,f,x)\otimes(n^\prime,f^\prime,x^\prime)\Big) =\tfrac{(-1)^{\vert x\vert +1}}{2}\,\Big(
\tau_{\mathsf{t}f}\Big(\LLL(f)x\otimes \Xi_{(n^\prime,f^\prime)}^{(n,f)} x^\prime\Big) -
\tau_{\mathsf{t}f^\prime}\Big(\Xi_{(n,f)}^{(n^\prime,f^\prime)}x\otimes \LLL(f^\prime)x^\prime\Big)
\Big)~~,
$}
\end{flalign}
\end{subequations}
for all $(n,f,x),(n^\prime,f^\prime,x^\prime)\in \widetilde{\LLL}_{1,\bullet}$.
\begin{propo}\label{prop:homotopyPoisson}
The family $\{\tau_{\bbL,n} \in 
\hom((\widetilde{\LLL}\wedge\widetilde{\LLL})_{n,\bullet} , \bbK)_n \}_{n\in\bbZ}$ 
given by \eqref{eqn:taubbL} defines a Poisson structure on $\bbL L_!(\LLL) = \Tot(\widetilde{\LLL})
: \BB\bbZ\to \Ch_{\bbK}$.
\end{propo}
\begin{proof}
Invariance under the $\bbZ$-action follows directly from \eqref{eqn:translationinvariancezigzagging}
and the explicit formulas for the $\Lambda$, $\Gamma$ and $\Xi$ 
given in \eqref{eqn:Hhformula}, \eqref{eqn:Hkformula}
and \eqref{eqn:Kformula}.
The tower of conditions \eqref{eqn:bicomplexPoisson} reduces in the present case to
\begin{flalign}\label{eqn:tmpconditions}
\tau_{\bbL,0}\circ \dd^{\widetilde{\LLL}\wedge\widetilde{\LLL}}=0~~, \quad
\tau_{\bbL,0} \circ \delta^{\widetilde{\LLL}\wedge\widetilde{\LLL}} + \tau_{\bbL,1}\circ  \dd^{\widetilde{\LLL}\wedge\widetilde{\LLL}}=0~~,\quad
\tau_{\bbL,1}\circ \delta^{\widetilde{\LLL}\wedge\widetilde{\LLL}} + \tau_{\bbL,2}\circ  \dd^{\widetilde{\LLL}\wedge\widetilde{\LLL}}=0\quad.
\end{flalign}
The first condition is immediate because $\tau_{\bbL,0}$ given in \eqref{eqn:taubbL} is clearly a chain map. 
By slightly lengthy but straightforward calculations, the second condition 
follows from \eqref{eqn:zigzagginghomotopyh} 
and \eqref{eqn:zigzagginghomotopyk}, and the third condition follows from \eqref{eqn:zigzagging2homotopy}.
({\em Note:} Because the $\tau_{\bbL,n}$ are obtained by antisymmetrization with respect to the total degrees
and the conditions \eqref{eqn:tmpconditions} are invariant under the symmetric braiding in $\bCh_\bbK$,
it is sufficient to consider in these calculations only the first of the two terms in \eqref{eqn:taubbL}.)
\end{proof}

We shall now construct the components
$\rho_N\in \hom(Q(\LLL)(N)\wedge Q(\LLL)(N),\bbK)_1$ of the natural chain homotopy $\rho$
from \eqref{eqn:PBPoissonhomtopy}, for all $N\in\CC$. Using that $Q(\LLL)(N) = \Tot(\LLL^\Delta(N))$ 
is the $\bigoplus$-totalization of the bicomplex \eqref{eqn:Xdelta},
we can use again a direct sum decomposition into the vertical degrees of 
$\LLL^\Delta(N)$ to define $\rho_N$ by a family of $n+1$-chains
\begin{subequations}\label{eqn:rhoNrelations}
\begin{flalign}
\Big\{\rho_{N,n} \in \hom\big((\LLL^{\Delta}(N)\wedge \LLL^\Delta(N))_{n,\bullet},\bbK\big)_{n+1}\Big\}_{n\in\bbZ}
\end{flalign}
that satisfy
\begin{flalign}
\tau_{\bbL,n}\circ (\eta_\LLL\wedge\eta_\LLL) - \tau_{N,n}\circ(q_\LLL\wedge q_\LLL) = \rho_{N,n-1}\circ \delta^{\LLL^{\Delta}(N)\wedge \LLL^\Delta(N)} + \rho_{N,n}\circ \dd^{\LLL^{\Delta}(N)\wedge \LLL^\Delta(N)}\quad,
\end{flalign}
\end{subequations}
where $\tau_{N,n}$ denotes the decomposition of the original Poisson 
structure $\tau_N : \LLL(N)\wedge\LLL(N)\to\bbK$
into vertical degrees, i.e.\ $\tau_{N,0}=\tau_N$ and $\tau_{N,n}=0$ for all $n\neq 0$.
Here, we have used the formula $\partial\rho_N = 
\dd^\mathbb{K}\circ\rho_N+ \rho_N\circ \dd^{Q(\LLL)(N)}= \rho_N\circ \dd^{ \Tot(\LLL^\Delta(N))}$. 
We define
\begin{subequations}\label{eqn:rhoNexplicit}
\begin{flalign}
\rho_{N,n}=0~~,\quad\text{for all }n\not\in\{0,1\}\quad.
\end{flalign}
The degree $0$ component is defined by
\begin{flalign}
\resizebox{.9 \textwidth}{!} 
{$
\rho_{N,0}\Big((g,x)\otimes(g^\prime,x^\prime)\Big) = \tfrac{1}{2} \Big((-1)^{\vert x\vert}\,\tau_{N}\Big(\LLL(g)x\otimes \Lambda^{(0,g)}_{(L(g^\prime),\mathsf{s}g^\prime)}x^\prime\Big)+ \tau_N\Big(\Lambda_{(L(g),\mathsf{s}g)}^{(0,g^\prime)} x \otimes \LLL(g^\prime) x^\prime\Big)\Big)~~,
$}
\end{flalign}
for all $(g,x),(g^\prime,x^\prime)\in \LLL^{\Delta}(N)_{0,\bullet}$,
and the degree $1$ component is defined by
\begin{flalign}
\rho_{N,1}\Big((g,f,x)\otimes(g^\prime,x^\prime)\Big) = -\tfrac{(-1)^{\vert x^\prime\vert}}{2}\,\tau_N\Big(\Xi^{(0,g^\prime)}_{(0,f)}x\otimes \LLL(g^\prime) x^\prime\Big) \quad,
\end{flalign}
for all $(g,f,x)\in  \LLL^{\Delta}(N)_{1,\bullet}$ and $(g^\prime,x^\prime)\in 
\LLL^{\Delta}(N)_{0,\bullet}$, and
\begin{flalign}
\rho_{N,1}\Big((g,x)\otimes(g^\prime,f^\prime, x^\prime)\Big) 
= -\tfrac{1}{2} \,\tau_N\Big(\LLL(g)x\otimes \Xi^{(0,g)}_{(0,f^\prime)}x^\prime\Big)\quad,
\end{flalign}
\end{subequations}
for all $(g,x)\in  \LLL^{\Delta}(N)_{0,\bullet}$ and $(g^\prime,f^\prime,x^\prime)\in 
\LLL^{\Delta}(N)_{1,\bullet}$. (Recall that, due to the fact that the nerve of $\CC$ is degenerate
in degrees $\geq 2$, we necessarily have that $(g,f,x) = (\id_N,f,x)$
and $(g^\prime,f^\prime,x^\prime)= (\id_N,f^\prime,x^\prime)$ in these expressions.)
\begin{propo}\label{prop:Poissonhomotopy}
The components \eqref{eqn:rhoNexplicit}
define a natural chain homotopy 
$\rho\in \hom(Q(\LLL)\wedge Q(\LLL),\bbK)_1$ satisfying \eqref{eqn:PBPoissonhomtopy}.
\end{propo}
\begin{proof}
This is a straightforward check using the properties \eqref{eqn:zigzagginghomotopyh},
\eqref{eqn:zigzagginghomotopyk} and \eqref{eqn:zigzagging2homotopy}
of the zig-zagging homotopies and naturality of the original Poisson structure
$\tau:\LLL\wedge\LLL\to \bbK$.
\end{proof}

Let us recall from \cite[Section 5]{LinearYM} the functor
\begin{flalign}
\CCR\,:\, \PoCh_\bbK~\longrightarrow~\dgAlg_\bbK
\end{flalign}
that assigns to a Poisson chain complex $(V,\tau)\in\PoCh_\bbK$
its associated differential graded CCR-algebra 
$\CCR(V,\tau)\in \dgAlg_\bbK$. 
Composing this functor with our given
functor $(\LLL,\tau) : \CC\to\PoCh_\bbK$ that models
the Poisson chain complexes of linear observables of a 
linear homotopy AQFT, we obtain a homotopy AQFT 
\begin{flalign}\label{eqn:originalAQFT}
\AAA := \CCR(\LLL,\tau) \,:\, \CC~\longrightarrow~\dgAlg_\bbK
\end{flalign}
on the category $\CC$ given in \eqref{eqn:RCEcategory}.
Furthermore, composing $\CCR$ with the pullback of the functor
$(\bbL L_!(\LLL),\tau_{\bbL}) : \BB\bbZ\to\PoCh_\bbK$ 
along the localization functor $L : \CC\to\BB\bbZ$ defines
another homotopy AQFT 
\begin{flalign}\label{eqn:strictifiedAQFT}
\AAA^{\mathrm{st}} := \CCR\big(L^\ast \bbL L_!(\LLL),L^\ast\tau_{\bbL}\big)
= L^\ast\big(\CCR\big(\bbL L_!(\LLL),\tau_{\bbL}\big)\big)\,:\, \CC~\longrightarrow~\dgAlg_\bbK
\end{flalign}
on $\CC$ which, due to the fact that it is obtained as a pullback along the localization functor,
satisfies the strict time-slice axiom.
The following rectification theorem is the main result of this paper.
\begin{theo}\label{theo:rectificationhomotopyAQFT}
Let $(\LLL,\tau) : \CC\to\PoCh_\bbK$ be any functor on the category $\CC$ in \eqref{eqn:RCEcategory}
that satisfies the homotopy time-slice axiom, i.e.\ $\LLL\in \Fun(\CC,\Ch_\bbK)^{\mathrm{ho}\mathrm{All}}$.
Then its associated homotopy AQFT \eqref{eqn:originalAQFT}, which satisfies the homotopy time-slice axiom,
is equivalent via a zig-zag of weak equivalences in the model category $\Fun(\CC,\dgAlg_\bbK)$
to the homotopy AQFT \eqref{eqn:strictifiedAQFT} that satisfies the strict time-slice axiom.
\end{theo}
\begin{proof}
From the zig-zag of weak equivalences in \eqref{eqn:zigzagWE}
and \cite[Proposition 5.3.\. (a)]{LinearYM}, we obtain the following two
weak equivalences
\begin{flalign}\label{eqn:tmpWE}
\xymatrix@C=1.5em{
\CCR(q_{\LLL})\,:\, \CCR\big(Q(\LLL), q_{\LLL}^\ast(\tau) \big) \ar[r]^-{\sim}~&~\AAA
}~~,\quad 
\xymatrix@C=1.5em{
\CCR(\eta_{\LLL}) \,:\, \CCR\big(Q(\LLL),\eta_{\LLL}^\ast(L^\ast \tau_\bbL)\big)~\ar[r]^-{\sim}&~\AAA^{\mathrm{st}}
}
\end{flalign}
in $\Fun(\CC,\dgAlg_\bbK)$. Using that 
$q_{\LLL}^\ast(\tau)$ and $\eta_{\LLL}^\ast(L^\ast \tau_\bbL) = q_{\LLL}^\ast(\tau) + \partial\rho$
are homotopic by Proposition \ref{prop:Poissonhomotopy}, we obtain from \cite[Proposition 5.3.\. (b)]{LinearYM}
a zig-zag of weak equivalences
\begin{flalign}\label{eqn:tmpWE2}
\xymatrix@C=1.5em{
\CCR\big(Q(\LLL), q_{\LLL}^\ast(\tau) \big) ~&~\ar[l]_-{\sim} A_{(Q(\LLL),q_{\LLL}^\ast(\tau) ,\rho)} \ar[r]^-{\sim}
~&~ \CCR\big(Q(\LLL),\eta_{\LLL}^\ast(L^\ast \tau_\bbL)\big)
}
\end{flalign}
in $\Fun(\CC,\dgAlg_\bbK)$, where the interpolating 
object $ A_{(Q(\LLL),q_{\LLL}^\ast(\tau) ,\rho)}\in \Fun(\CC,\dgAlg_\bbK)$
was constructed  explicitly in \cite{LinearYM}. Combining \eqref{eqn:tmpWE} and \eqref{eqn:tmpWE2}
completes the proof.
\end{proof}

We conclude by emphasizing that the relative Cauchy evolution for the
strictified homotopy AQFT \eqref{eqn:strictifiedAQFT}
is given by applying the $\CCR$-functor on the RCE automorphism 
for linear observables from \eqref{eqn:smallRCEmap}, i.e.\
\begin{flalign}\label{eqn:bigRCEmap}
\RCE_{(M,h)}^{} \,:=\,\CCR(\rce_{(M,h)}^{})\,:\, \CCR\big(\bbL L_!(\LLL),\tau_\bbL\big)~\longrightarrow~
\CCR\big(\bbL L_!(\LLL),\tau_\bbL\big)\quad.
\end{flalign}
In particular, the relative Cauchy evolution 
of the strictified model is given by a {\em strict  automorphism}
of dg-algebras, in contrast to an $A_\infty$-quasi-automorphism.

\begin{rem}\label{rem:involution}
Throughout the bulk of this paper we have neglected $\ast$-involutions
on the homotopy AQFTs in order to streamline our presentation.
However, we would like to emphasize that our main result in Theorem \ref{theo:rectificationhomotopyAQFT}
also holds true for homotopy AQFTs with $\ast$-involutions.
The proof is obtained by choosing 
up to Proposition \ref{prop:Poissonhomotopy} the ground field $\bbK=\bbR$
and recalling that the homotopical
properties of the $\CCR$-functor entering 
the proof of Theorem \ref{theo:rectificationhomotopyAQFT}
were shown in \cite[Proposition 5.3]{LinearYM} also for the 
complexified $\CCR$-functor with $\ast$-involutions over $\bbC$.
\end{rem}


\section{\label{sec:Example}Example: Linear quantum Yang-Mills theory}
The aim of this section is to illustrate our construction of the relative Cauchy 
evolution for linear homotopy AQFTs through a simple example,
the linear Yang-Mills model presented in \cite{LinearYM}.
This example is based on the solution complexes
\begin{flalign}\label{eqn:solcomplex}
\Sol(M) \,:=\, \Big(
\xymatrix@C=2.5em{
\stackrel{(-2)}{\Omega^0(M)} 
& \ar[l]_-{\delta_M} \stackrel{(-1)}{\Omega^1(M)} 
& \ar[l]_-{\delta_M\dd_{M} }\stackrel{(0)}{\Omega^1(M)} 
& \ar[l]_-{\dd_{M}} \stackrel{(1)}{\Omega^0(M)}
}
\Big)\,\in\,\Ch_\bbR\quad,
\end{flalign}
for each spacetime $M\in\Loc$, where the round brackets indicate homological degrees,
$\dd_M$ is the de Rham differential and $\delta_M$ the codifferential.
To be explicit, we work in arbitrary spacetime dimension and our convention 
for the metric signature is $(+--\cdots -)$. The Hodge operator $\ast_M^{}$ is defined so that
\begin{flalign}
\omega\wedge\ast_M^{} \eta \,=\, \frac{1}{p!} \omega_{a_1\cdots a_p} \eta^{a_1\cdots a_p}\vol_M\quad,  
\end{flalign}
for $p$-forms $\omega,\eta\in \Omega^p(M)$, where $\vol_M=\ast_M(1)$ is 
the metric-induced volume form and we have used abstract index notation when 
regarding $p$-forms as anti-symmetric tensors. The codifferential $\delta_M$ 
is the adjoint of $\dd_M$ with respect to the pairing
\begin{flalign}
\langle \omega, \eta\rangle_M^{} \,=\, \int_M\omega\wedge \ast_M^{} \eta \quad,
\end{flalign}
which is defined when the supports of $\omega$ and $\eta$ have compact intersection.
These conventions coincide with those of \cite{FewsterLang} and were implicit in \cite{LinearYM}.
\sk

Elements $A\in \Sol(M)_0=\Omega^1(M)$ in degree $0$
are interpreted as gauge fields, elements $c\in \Sol(M)_1=\Omega^0(M)$ 
in degree $1$ as ghost fields, and elements $A^\ddagger\in \Sol(M)_{-1}=\Omega^1(M)$ and
$c^\ddagger\in\Sol(M)_{-2}=\Omega^0(M)$ in negative degrees as the antifields
of the theory. Note that the zeroth homology of the chain complex
\eqref{eqn:solcomplex} is the vector space of gauge equivalence classes of solutions
to Maxwell's equations, i.e.\ $H_0(\Sol(M)) =
 \{A\in \Omega^1(M)\,:\, \delta_M\dd_M A=0\}/\dd_M\Omega^0(M)$.
The homologies of $\Sol(M)$ in non-zero degrees carry refined information
about the gauge symmetries and dynamics of linear Yang-Mills theory, see
\cite[Example 3.9]{LinearYM} for a discussion.
\sk

The chain complex of linear observables\footnote{The terminology is chosen 
for simplicity and consistency with~\cite{LinearYM}; in other contexts some 
elements would be described as smearings of unobservable fields.} on $M\in\Loc$ for this theory 
is given by the smooth dual of the solution complex \eqref{eqn:solcomplex}, i.e.\
\begin{flalign}\label{eqn:linearobscomplex}
\LLL(M) \,:=\, \Big(
\xymatrix@C=2.5em{
\stackrel{(-1)}{\Omega^0_\cc(M)} 
& \ar[l]_-{-\delta_M} \stackrel{(0)}{\Omega^1_\cc(M)} 
& \ar[l]_-{\delta_M\dd_M}\stackrel{(1)}{\Omega^1_\cc(M)} 
& \ar[l]_-{-\dd_M} \stackrel{(2)}{\Omega^0_\cc(M)}
}
\Big)\,\in\,\Ch_\bbR\quad,
\end{flalign}
where the subscript ${}_\cc$ denotes differential forms of compact support.
Elements $\varphi\in \LLL(M)_0=\Omega^1_\cc(M)$ in degree $0$ are interpreted
as linear gauge field observables, elements $\chi\in \LLL(M)_{-1}=\Omega^0_\cc(M)$ 
in degree $-1$ as linear ghost field observables, and elements
$\alpha\in \LLL(M)_{1}=\Omega^1_\cc(M)$ and $\beta\in \LLL(M)_2=\Omega^0_\cc(M)$
in positive degrees as linear observables for the antifields.
The evaluation pairing 
\begin{subequations}\label{eqn:pairing}
\begin{flalign}
\langle-,-\rangle_M^{} \,:\,\LLL(M)\otimes\Sol(M)~\longrightarrow~\bbR
\end{flalign}
between linear observables and the fields in \eqref{eqn:solcomplex} is the chain map defined by
\begin{flalign}
\langle \varphi, A \rangle_M^{} &= \int_M \varphi \wedge \ast_M A ~~\quad,\qquad
\langle \chi , c \rangle_M^{} = \int_M \chi \, c\, \vol_M \quad, \qquad\\
\langle \alpha , A^\ddagger \rangle_M^{} &=\int_M \alpha \wedge \ast_M A^\ddagger \quad,\qquad
\langle \beta, c^\ddagger \rangle_M^{} = \int_M \beta \, c^\ddagger \,\vol_M \quad.
\end{flalign}
\end{subequations} 
\sk

The Poisson structure $\tau_M : \LLL(M)\wedge\LLL(M)\to\bbR$ on \eqref{eqn:linearobscomplex}
is constructed by chain complex analogs of retarded/advanced Green operators, which were
called retarded/advanced trivializations in \cite{LinearYM} and were shown 
to be unique up to contractible choices.
These structures provide chain homotopies trivializing the inclusion chain map
\begin{flalign}\label{eqn:jmapping}
\parbox{0.3cm}{\xymatrix{
\LLL(M)\ar[d]_-{j_M}\\
\Sol(M)[1]
}
}:=~~ \left(\parbox{2cm}{\xymatrix@C=2.5em{
\ar[d]_-{-\subseteq}\Omega^0_\cc(M)
& \ar[d]_-{-\subseteq}\ar[l]_-{-\delta_M} \Omega^1_\cc(M)
& \ar[d]_-{\subseteq} \ar[l]_-{\delta_M\dd_M} \Omega^1_\cc(M)
&\ar[d]_-{\subseteq} \ar[l]_-{-\dd_M} \Omega^0_\cc(M)
\\ 
\Omega^0(M)
& \ar[l]^-{-\delta_M} \Omega^1(M)
& \ar[l]^-{-\delta_M\dd_{M} } \Omega^1(M)
& \ar[l]^-{-\dd_{M}} \Omega^0(M)
}
}\right)
\end{flalign}
from the chain complex of linear observables to the shifted solution complex.
(In \eqref{eqn:jmapping} we have chosen sign conventions
for the vertical inclusion maps that are opposite to the ones in \cite{LinearYM}.
Our present choice is convenient to fix a minor sign error in Eqn.~(4.26) of the latter paper.)
A concrete model for such retarded/advanced Green homotopies 
\begin{subequations}\label{eqn:Greenhomotopies}
\begin{flalign}
\mathcal{G}^\pm_M \,\in\, \hom\big(\LLL(M),\Sol(M)[1]\big)_1
\end{flalign}
satisfying
\begin{flalign}\label{eqn:Greenhomotopies2}
j_M \,=\, \partial_M \mathcal{G}^\pm_M \,=\, \dd^{\Sol(M)[1]}\, \mathcal{G}^\pm_M  + \mathcal{G}^\pm_M \,\dd^{\LLL(M)}
\end{flalign}
can be given in terms of the retarded/advanced Green operators $G_{M}^\pm$
of the d'Alembert operator $\square_M := \delta_M\dd_M+\dd_M\delta_M$
on differential forms,\footnote{This convention for $\Box_M$ follows~\cite{LinearYM} but 
is opposite in sign to that adopted in~\cite{FewsterLang}; accordingly, 
the Green operators in these two references differ by an overall sign.} reading explicitly as
\begin{flalign}\label{eqn:Greenhomotopies3}
\mathcal{G}^\pm_M (\chi ) \,:=\, G_{M}^{\pm} (\dd_M\chi)~,~~
\mathcal{G}^\pm_M (\varphi )\,:=\, G_M^{\pm}(\varphi)~,~~
\mathcal{G}^\pm_M (\alpha )\,:=\, -G_M^{\pm}( \delta_M \alpha)~,~~
\mathcal{G}^\pm_M (\beta )\,:=\,0 \quad,
\end{flalign}
\end{subequations}
for all $\chi\in \LLL(M)_{-1}$, $\varphi\in \LLL(M)_0$,
$\alpha\in \LLL(M)_{1}$ and $\beta\in \LLL(M)_{2}$.
Taking the difference 
\begin{flalign}
\mathcal{G}_M\,:=\, \mathcal{G}_M^{+} - \mathcal{G}_M^-\,:\,\LLL(M)~\longrightarrow~\Sol(M)
\end{flalign}
between the retarded and advanced Green homotopy defines a chain map
to the unshifted solution complex, which plays a similar role to the 
causal propagator in ordinary AQFT. Explicitly, this chain map reads as
\begin{flalign}\label{eqn:causalpropagator}
\parbox{0.3cm}{\xymatrix{
\LLL(M)\ar[d]_-{\mathcal{G}_M}\\
\Sol(M)
}
} =~~ \left(\parbox{2cm}{\xymatrix@C=2.5em{
\ar[d]_-{0}0
&\ar[l]_-{0}\ar[d]_-{G_M\dd_M}\Omega^0_\cc(M)
& \ar[d]_-{G_M}\ar[l]_-{-\delta_M} \Omega^1_\cc(M)
& \ar[d]_-{-G_M\delta_M} \ar[l]_-{\delta_M\dd_M} \Omega^1_\cc(M)
&\ar[d]_-{0} \ar[l]_-{-\dd_M} \Omega^0_\cc(M)
\\ 
\Omega^0(M)
& \ar[l]^-{\delta_M} \Omega^1(M)
& \ar[l]^-{\delta_M\dd_{M} } \Omega^1(M)
& \ar[l]^-{\dd_{M}} \Omega^0(M)
& \ar[l]^-{0}0
}
}\right)\quad,
\end{flalign}
where $G_M := G_M^+ - G_M^-$ denotes the causal propagator for the d'Alembert 
operator on differential forms. The Poisson structure is then defined
by the chain map
\begin{subequations}\label{eqn:PoissonYM}
\begin{flalign}
\tau_M\,:=\, \langle-,\mathcal{G}_M(-)\rangle_M^{}\,:\, \LLL(M)\wedge\LLL(M)~\longrightarrow~\bbR\quad,
\end{flalign}
whose non-zero components explicitly read as
\begin{flalign}
\tau_M(\varphi_1\otimes \varphi_2) &\,=\, \int_M \varphi_1\wedge \ast_M G_M (\varphi_2) \,=\, -\tau_M(\varphi_2\otimes \varphi_1)\quad,\\
\tau_M(\alpha\otimes\chi) &\,=\, \int_M\alpha\wedge \ast_M G_M(\dd_M\chi) \,=\, \tau_M(\chi\otimes\alpha)  \quad,
\end{flalign}
\end{subequations}
for all $\chi\in \LLL(M)_{-1}$, $\varphi_1,\varphi_2 \in \LLL(M)_0$ and
$\alpha\in \LLL(M)_{1}$. (As already mentioned above, there is a minor 
sign error in the corresponding equation (4.26) in \cite{LinearYM}. Our opposite
sign convention for the vertical maps in \eqref{eqn:jmapping} implies that our 
Poisson structure \eqref{eqn:PoissonYM} has a positive sign for linear observables in degree $0$,
in contrast to the negative sign in the corrected version of Eqn.\ (4.26a) in \cite{LinearYM}.)
\sk

Naturality of \eqref{eqn:causalpropagator} and \eqref{eqn:pairing}
implies that the assignment $M\mapsto (\LLL(M),\tau_M)$ can be promoted to 
a functor $(\LLL,\tau) : \Loc\to\PoCh_\bbR$. To a $\Loc$-morphism
$f:M\to N$, this functor assigns the chain map
$\LLL(f) := f_\ast : (\LLL(M),\tau_M)\to(\LLL(N),\tau_N)$ of Poisson chain complexes 
given by pushforward (i.e.\ extension by zero) of compactly supported differential forms.
By \cite[Theorem 6.19]{LinearYM}, the composition of this functor
with the $\CCR$-functor $\CCR : \PoCh_\bbR\to \astdgAlg_\bbC$ 
(in the context of chain complexes and with $\ast$-involutions) defines a
homotopy AQFT $\AAA :=\CCR(\LLL,\tau)\in {}^\ast\AQFT_\infty(\overline{\Loc})^{\mathrm{ho}W}$
with $\ast$-involution that satisfies the homotopy time-slice axiom.
With an abuse of notation, we denote by the same symbol
$\AAA: \CC\to \astdgAlg_\bbC$ the restriction of this homotopy AQFT 
to the category $\CC$ in \eqref{eqn:RCEcategory} that is relevant for the discussion
of the relative Cauchy evolution induced by the pair $(M,h)$
consisting of a spacetime $M$ and a compactly supported metric perturbation $h$.
\sk

Our main Theorem \ref{theo:rectificationhomotopyAQFT}, see also Remark \ref{rem:involution}, 
provides us with a concept of relative Cauchy evolution for this homotopy AQFT. Let us recall that 
our construction is slightly abstract as it involves a replacement of the theory $\AAA$
by a weakly equivalent theory $\AAA^\mathrm{st}: \CC\to \astdgAlg_\bbC$ (see \eqref{eqn:strictifiedAQFT})
that satisfies the {\em strict} time-slice axiom and hence admits an RCE automorphism
in the ordinary sense \eqref{eqn:bigRCEmap}. In order to physically interpret the resulting concept
of relative Cauchy evolution, and to compute its associated stress-energy tensor, we make
use of a concept similar to the quantum fields in the sense of \cite{BFV,FewsterVerch}, 
where they are understood as natural transformations 
between the functors assigning test function spaces and field algebras to spacetimes. 
We shall use the chain complex $\LLL(M)$ of linear observables on the object $M$ in \eqref{eqn:RCEcategory}
for labeling these quantum fields and consider the chain map
\begin{flalign}\label{eqn:iotamap}
\iota\,:\, \LLL(M) ~\longrightarrow~\bbL L_!(\LLL)~,~~\omega~\longmapsto ~ (0,M,\omega)
\end{flalign}
to the chain complex of linear observables for the strictified theory given in
\eqref{eqn:leftderivedfunctor} and \eqref{eqn:tildeXbicomplex}. In words,
we embed the linear observables on $M$ at the level $n=0$ in the spiral 
\eqref{eqn:spiral} controlling our strictification construction. (Due to $\bbZ$-invariance
of $\bbL L_!(\LLL)$ and its Poisson structure $\tau_\bbL$, it does not matter 
at which level $n\in\bbZ$ we embed $\LLL(M)$.)
Using further the chain map $\bbL L_!(\LLL)\to \CCR(\bbL L_!(\LLL),\tau_\bbL)$
that assigns to linear observables their corresponding generators in the CCR-algebra,
we can compose with \eqref{eqn:iotamap} to obtain a chain map from
$\LLL(M)$ to $\CCR(\bbL L_!(\LLL),\tau_\bbL)$ that plays a similar role as the quantum fields
in \cite{BFV,FewsterVerch}.  We can now introduce a concept of relative Cauchy evolution 
for this particular quantum field, which can be interpreted as a combination of the gauge field, ghost
and antifields into a single multicomponent field.
\begin{propo}\label{propo:rcequantumfield}
The diagram
\begin{flalign}
\xymatrix@C=4em{
\ar[d]_-{Z_{(1,M)}^{(0,M)}} \LLL(M) \ar[r]^-{\iota} ~&~ \ar[d]^-{\rce_{(M,h)}^{}} \ar[r]^-{} \bbL L_!(\LLL) \ar@{<=}[dl]_-{\theta}~&~ \ar[d]^-{\RCE_{(M,h)}^{}}\CCR(\bbL L_!(\LLL),\tau_\bbL)\\
\LLL(M) \ar[r]_-{\iota}~&~ \bbL L_!(\LLL) \ar[r] ~&~\CCR(\bbL L_!(\LLL),\tau_\bbL)
}
\end{flalign}
in $\Ch_\bbR$ is homotopy commutative, where 
\begin{flalign}\label{eqn:zigzagrceexplicit}
Z_{(1,M)}^{(0,M)} \,=\, \LLL(i_-)\,\LLL(j_-)^{-1}\,\LLL(j_+)\,\LLL(i_+)^{-1} \,:\, \LLL(M) ~\longrightarrow~ \LLL(M)
\end{flalign}
denotes the zig-zagging chain map from \eqref{eqn:zigzaggingmaps} and 
the chain homotopy $\theta\in \hom(\LLL(M),\bbL L_!(\LLL))_1$
is constructed in the proof.
\end{propo}
\begin{proof}
The right square commutes (strictly) by the definition of the RCE automorphism 
on CCR-algebras, see \eqref{eqn:bigRCEmap}.
Homotopy commutativity of the left square means that there
exists $\theta\in \hom(\LLL(M),\bbL L_!(\LLL))_1$ such that
$\rce_{(M,h)}^{}\,\iota -\iota\,Z_{(1,M)}^{(0,M)}= \partial \theta
= (\delta^{\widetilde{\LLL}}+\dd^{\widetilde{\LLL}}) \theta + \theta\,\dd^{\LLL(M)}$,
where we recall that
$\bbL L_!(\LLL)=\Tot(\widetilde{\LLL})$ is obtained as the $\bigoplus$-totalization 
of a bicomplex \eqref{eqn:tildeXbicomplex}
whose horizontal and vertical differentials are denoted here by $\dd^{\widetilde{\LLL}}$ 
and $\delta^{\widetilde{\LLL}}$. Evaluating on any $x \in\LLL(M)$ and recalling the RCE automorphism
on linear observables \eqref{eqn:smallRCEmap}, the relevant homotopy relation is given by
\begin{flalign}\label{eqn:lambdahomotopy}
\big(1,M,x\big)- \big(0,M,Z_{(1,M)}^{(0,M)}x\big) \,=\,(\delta^{\widetilde{\LLL}}+\dd^{\widetilde{\LLL}}) \theta(x) + \theta(\dd^{\LLL(M)} x)\quad. 
\end{flalign}
We construct the chain homotopy $\theta$ by transporting the element
$(1,M,x)\in \bbL L_!(\LLL)$ step by step from right to left along the relevant part
\begin{flalign}\label{eqn:shortzigzag}
\xymatrix{
(0,M) &\ar[l]_-{i_-} (1,M_-) \ar[r]^-{j_-}& (1,M_h)&\ar[l]_-{j_+} (1,M_+) \ar[r]^-{i_+} & (1,M)
}
\end{flalign}
of the zig-zags in \eqref{eqn:spiral}. 
\sk

To describe the chain homotopies for
the individual steps, let $f : (n+L(f) , \mathsf{s}f)\to (n,\mathsf{t}f)$
be any of the morphisms in \eqref{eqn:shortzigzag}. 
If $f$ points {\em from right to left}
in \eqref{eqn:shortzigzag}, we define the homotopy
$\theta_{\overleftarrow{f}}\in \hom(\LLL(\mathsf{s}f),\bbL L_!(\LLL))_1$ by
\begin{subequations}
\begin{flalign}
\theta_{\overleftarrow{f}}(x) \,:=\, (-1)^{\vert x\vert} \,\big(n,f,x\big)\quad,
\end{flalign}
which satisfies
\begin{flalign}
\partial \theta_{\overleftarrow{f}}(x) = \big(n+L(f),\mathsf{s}f,x\big) -\big(n,\mathsf{t} f, \LLL(f)x\big)\quad,
\end{flalign}
\end{subequations}
for all $x\in \LLL(\mathsf{s}f)$. (Observe that this homotopy moves the element
$(n+L(f),\mathsf{s}f,x)$ from right to left in \eqref{eqn:shortzigzag}.)
If $f$ points {\em from left to right}
in \eqref{eqn:shortzigzag}, we define the homotopy
$\theta_{\overrightarrow{f}}\in \hom(\LLL(\mathsf{t}f),\bbL L_!(\LLL))_1$ by
\begin{subequations}
\begin{flalign}
\theta_{\overrightarrow{f}}(x) \,:=\, - (-1)^{\vert x\vert} \,\big(n,f,\LLL(f)^{-1} x\big) -\big(n,\mathsf{t}f, \lambda_f x\big)\quad,
\end{flalign}
where $\LLL(f)^{-1}$  is a quasi-inverse of $\LLL(f)$ and $\lambda_f$ its associated homotopy coherence 
data from \eqref{eqn:homotopy4quasiinverse}, which satisfies
\begin{flalign}
\partial \theta_{\overrightarrow{f}}(x) = \big(n,\mathsf{t}f,x\big) -\big(n+L(f),\mathsf{s} f, \LLL(f)^{-1}x\big)\quad,
\end{flalign}
\end{subequations}
for all $x\in \LLL(\mathsf{t}f)$.  (Observe that this homotopy moves the element
$(n,\mathsf{t}f,x)$ from right to left in \eqref{eqn:shortzigzag}.)
Composing these basic homotopies defines the required chain homotopy in
\eqref{eqn:lambdahomotopy}.
\end{proof}
\begin{rem}
Because $\iota :  \LLL(M)\to \bbL L_!(\LLL)$ is a quasi-isomorphism,
our model given in Proposition \ref{propo:rcequantumfield}
for the relative Cauchy evolution for this linear quantum field is unique up to homotopy.
Note further that this model agrees with the ``naive approach'' consisting of quasi-inverting
the quasi-isomorphisms associated with the Cauchy morphisms in \eqref{eqn:RCEcategory} 
since the relevant zig-zagging chain map \eqref{eqn:zigzagrceexplicit}
is of the same form as the usual relative Cauchy evolution formula in \eqref{eqn:RCEA}.
It is important to stress that this ``naive approach'' works only at the level of 
linear quantum fields because the chain map \eqref{eqn:zigzagrceexplicit} does in general {\em not}
preserve the Poisson structure $\tau_M : \LLL(M)\wedge\LLL(M)\to \bbR$ 
and, as a consequence, it does {\em not} lift to the CCR-algebra $\CCR(\LLL(M),\tau_M)$. 
Hence, the role of our rectification Theorem \ref{theo:rectificationhomotopyAQFT}
is to establish {\em existence} of a well-behaved concept of relative Cauchy evolution for linear homotopy AQFTs,
while the role of Proposition \ref{propo:rcequantumfield} is to provide explicit {\em computational  tools}
to study the relative Cauchy evolution at the level of linear quantum fields.
\end{rem}

In order to obtain an explicit formula for the relative Cauchy evolution
for linear quantum fields determined by Proposition \ref{propo:rcequantumfield},
we have to construct explicitly the necessary quasi-inverses in \eqref{eqn:zigzagrceexplicit}
and their corresponding homotopy coherence data \eqref{eqn:homotopy4quasiinverse}.
Without much extra effort, we will perform this construction for a general
Cauchy morphism $f: N\to N^\prime$ in $\Loc$. Let us choose any two Cauchy surfaces
$\Sigma_\pm\subseteq N$ in the source spacetime $N$ such that
$\Sigma_+$ lies in the future of $\Sigma_-$, i.e.\ $\Sigma_+ \subseteq I^+_N(\Sigma_-)$.
We denote by $\Sigma_\pm^\prime := f(\Sigma_\pm)\subseteq N^\prime$
their images under $f$, which are Cauchy surfaces in $N^\prime$ because $f$ is by
hypothesis a Cauchy morphism. We then pick any partition of unity
$\rho^\prime_\pm$ subordinate to the open cover
$\{I^+_{N^\prime}(\Sigma^\prime_-), I^-_{N^\prime}(\Sigma^\prime_+)\}$ of $N^\prime$
and denote by $\rho_\pm := f^\ast(\rho_\pm^\prime)$ its pullback along $f$,
which defines a partition of unity subordinate to the open cover
$\{I^+_{N}(\Sigma_-), I^-_{N}(\Sigma_+)\}$ of $N$.
From this data and the retarded/advanced Green homotopies
in \eqref{eqn:Greenhomotopies}, we define the following 
two chain homotopies
\begin{subequations}\label{eqn:tildehomotopiesexplicit}
\begin{flalign}
\widetilde{\lambda}_f\,&:=\, - \rho_+^\prime\, \mathcal{G}_{N^\prime}^- - \rho_-^\prime\, \mathcal{G}_{N^\prime}^+\,\in\,\hom\big(\LLL(N^\prime),\Sol(N^\prime)[1]\big)_1\quad,\\
\widetilde{\gamma}_f\,&:=\,-  \rho_+\, \mathcal{G}_{N}^- - \rho_-\, \mathcal{G}_{N}^+\,\in\,\hom\big(\LLL(N),\Sol(N)[1]\big)_1\quad.
\end{flalign}
\end{subequations}
\begin{lem}\label{lem:homotopiesexplicit}
The chain homotopy $\widetilde{\lambda}_f \in\hom(\LLL(N^\prime),\Sol(N^\prime)[1])_1$ 
factors uniquely through the chain map
$j_{N^{\prime}} : \LLL(N^\prime)\to \Sol(N^\prime)[1]$ given in \eqref{eqn:jmapping}, i.e.\ there exists a unique
chain homotopy $\lambda_f\in \hom(\LLL(N^\prime),\LLL(N^\prime))_1$ such that 
\begin{flalign}
\widetilde{\lambda}_f \,=\, j_{N^\prime}\,\lambda_f\quad.
\end{flalign}
In complete analogy, the chain homotopy $\widetilde{\gamma}_f \in\hom(\LLL(N),\Sol(N)[1])_1$ 
factors uniquely through 
$j_{N} : \LLL(N)\to \Sol(N)[1]$, defining a unique
chain homotopy $\gamma_f\in \hom(\LLL(N),\LLL(N))_1$ such that 
\begin{flalign}\label{eqn:gammawithouttilde}
\widetilde{\gamma}_f  \,=\, j_{N}\,\gamma_f\quad.
\end{flalign}
\end{lem}
\begin{proof}
The two proofs are identical, hence we shall write out only the one for $\widetilde{\gamma}_f$.
First, let us note that uniqueness of $\gamma_f$ follows from the fact that the chain map $j_N$ 
in \eqref{eqn:jmapping} is degree-wise injective. Concerning existence, 
we have to show that the image 
\begin{flalign}\label{eqn:tmpsupport}
\widetilde{\gamma}_f(\omega) = - \rho_+\, \mathcal{G}_{N}^-(\omega) - \rho_-\, \mathcal{G}_{N}^+(\omega)
\in\Sol(N)[1]
\end{flalign} 
of every element  $\omega\in\LLL(N)$ is a compactly supported differential form.
Recalling \eqref{eqn:Greenhomotopies3}, the support of $ \mathcal{G}_{N}^\pm(\omega)$ 
is strictly past/future compact, i.e.\ $\supp(\mathcal{G}_{N}^\pm(\omega))\subseteq J^\pm_N(K)$
for a compact subset $K\subseteq N$, because $ \mathcal{G}_{N}^\pm$ is a composition
of a retarded/advanced Green operator and a differential operator. 
Furthermore, the support of $\rho_\pm$ is by construction past/future compact.
It then follows from \cite[Lemma 1.9]{Baer} that both terms in \eqref{eqn:tmpsupport} are compactly supported
differential forms, hence so is $\widetilde{\gamma}_f(\omega) $.
\end{proof}
\begin{rem}
For later reference, we record explicit expressions for
the two chain homotopies constructed in Lemma \ref{lem:homotopiesexplicit}.
The chain homotopy $\lambda_f\in\hom(\LLL(N^\prime),\LLL(N^\prime))_1$ is given by
\begin{subequations}\label{eqn:hfcomponents}
\begin{flalign}
\lambda_f\, =\, - \rho^\prime_+\, G_{N^\prime}^-\, Q_{N^\prime} - \rho^\prime_-\, G_{N^\prime}^+ \, Q_{N^\prime}\quad,
\end{flalign}
where $G_{N^\prime}^\pm$ denotes the retarded/advanced
Green operator for the d'Alembert operator on differential forms and
\begin{flalign}\label{eqn:Qmap}
Q_{N^\prime}(\chi^\prime) =-\dd_{N^\prime} \chi^\prime~~,\quad 
Q_{N^\prime}(\varphi^\prime)=\varphi^\prime~~,\quad
Q_{N^\prime}(\alpha^\prime) =-\delta_{N^\prime}\alpha^\prime~~,\quad
Q_{N^\prime}(\beta^\prime)= 0\quad,
\end{flalign}
\end{subequations}
for all $\chi^\prime\in \LLL(N^\prime)_{-1}$, $\varphi^\prime\in\LLL(N^\prime)_0$,
$\alpha^\prime\in\LLL(N^\prime)_{1}$ and $\beta^\prime\in\LLL(N^\prime)_2$.
The chain homotopy $\gamma_f\in\hom(\LLL(N),\LLL(N))_1$ is given by
\begin{flalign}\label{eqn:kfcomponents}
\gamma_f \,=\, - \rho_+\, G_{N}^- \,Q_{N} - \rho_-\, G_{N}^+ \,Q_{N}\quad,
\end{flalign}
with $Q_N$ defined analogously to $Q_{N^\prime}$ in \eqref{eqn:Qmap}.
\end{rem}

Consider now the chain map
$j_{N^\prime} + \partial_{N^\prime} \widetilde{\lambda}_f : \LLL(N^\prime)\to \Sol(N^\prime)[1]$.
It is easy to see that the image 
$j_{N^\prime}(\omega) + \partial_{N^\prime} \widetilde{\lambda}_f (\omega)\in \Sol(N^\prime)[1]$
of every element $\omega\in\LLL(N^\prime)$ is a differential form supported 
in the closed time-slab $J^-_{N^\prime}(\Sigma^\prime_+) 
\cap J^+_{N^\prime}(\Sigma^\prime_-)\subseteq f(N)\subseteq N^\prime$
that is contained in the image of $f$. Indeed, restricting this differential form to the open subset
$U_+ := I^+_{N^\prime}(\Sigma^\prime_+)\subseteq N^\prime$, we obtain
\begin{flalign}
\big(j_{N^\prime}(\omega) + \partial_{N^\prime}\widetilde{\lambda}_f (\omega)\big)\big\vert_{U_+}
= j_{N^\prime}(\omega)\big\vert_{U_+} - \partial_{N^\prime}\mathcal{G}_{N^\prime}^- (\omega)\big\vert_{U_+}
= 0\quad,
\end{flalign}
where in the first equality we have used that $\rho_+^\prime\vert_{U_+}=1$ and
$\rho_-^\prime\vert_{U_+}=0$, and in the second equality we have used \eqref{eqn:Greenhomotopies2}.
A similar argument shows that the restriction to $U_- := I^-_{N^\prime}(\Sigma^\prime_-)\subseteq N^\prime$
is zero too. In terms of the untilded chain homotopies from Lemma \ref{lem:homotopiesexplicit},
this determines a chain map
\begin{flalign}\label{eqn:fNmapping}
\id + \partial_{N^\prime} \lambda_f \,:\, \LLL(N^\prime) ~\longrightarrow~\LLL(f(N))
\end{flalign}
that takes values in the sub-chain complex $\LLL(f(N))\subseteq \LLL(N^\prime)$ 
of linear observables on the image of $f:N\to N^\prime$.
Post-composing with the pullback $f^\ast : \LLL(f(N))\to\LLL(N)$ of differential forms
along the isomorphism $f:N\to f(N)$, we obtain a chain map
\begin{flalign}\label{eqn:quasiinverseexplicit}
\LLL(f)^{-1}\,:=\, f^\ast \,\big(\id + \partial_{N^\prime} \lambda_f\big)\,:\, \LLL(N^\prime)~\longrightarrow~\LLL(N)\quad.
\end{flalign}
We will now prove that \eqref{eqn:quasiinverseexplicit}  is indeed a
quasi-inverse of $\LLL(f) : \LLL(N)\to\LLL(N^\prime)$ and that the chain homotopies
from Lemma \ref{lem:homotopiesexplicit} provide the necessary homotopy coherence data
in \eqref{eqn:homotopy4quasiinverse} with trivial $2$-homotopy $\xi_f=0$.
\begin{propo}
For every Cauchy morphism $f:N\to N^\prime$ in $\Loc$, 
the following three identities hold true
\begin{subequations}
\begin{flalign}
\LLL(f)\,\LLL(f)^{-1} -\id &\,= \, \partial_{N^\prime} \lambda_{f}\quad,\\
\LLL(f)^{-1}\,\LLL(f) - \id &\, =\, \partial_N \gamma_f\quad,\\
\LLL(f)\,\gamma_f  -\lambda_f\,\LLL(f) &\,=\, 0\quad,
\end{flalign}
\end{subequations}
where the chain map $\LLL(f)^{-1}$ is defined in \eqref{eqn:quasiinverseexplicit} and
the chain homotopies $\lambda_f$ and $\gamma_f$ are defined in Lemma \ref{lem:homotopiesexplicit}.
\end{propo}
\begin{proof}
The first identity is a simple check
\begin{flalign}
\LLL(f)\,\LLL(f)^{-1} -\id = f_\ast\, f^\ast\,\big(\id + \partial_{N^\prime} \lambda_f\big) - \id = \partial_{N^\prime} \lambda_f\quad,
\end{flalign}
where in the second step we have used that $f_\ast\,f^\ast =\id$ on $\LLL(f(N))$. 
Because the chain map $j_N$ is degree-wise injective,
the second identity is equivalent to $j_N \, \LLL(f)^{-1}\,\LLL(f) - j_N = \partial_N \widetilde{\gamma}_f$.
Using \eqref{eqn:gammawithouttilde} and the fact that $j_N$ is a chain map,
one checks the latter identity as follows
\begin{flalign}
j_N \, \LLL(f)^{-1}\,\LLL(f) - j_N = j_N\,f^\ast\,\big(\id+\partial_{N^\prime} \lambda_f\big)\,f_\ast - j_N
=\partial_N\big(f^\ast\, \widetilde{\lambda}_f \,f_\ast\big) = \partial_N\widetilde{\gamma}_f\quad,
\end{flalign}
where in the second step we have  also used that $f^\ast\, f_\ast =\id$ on $\LLL(N)$
and naturality of $j_N$ and $\partial_N$. The last step follows from \eqref{eqn:tildehomotopiesexplicit}
and naturality of the Green homotopies. Explicitly,
\begin{flalign}
f^\ast\, \widetilde{\lambda}_f \,f_\ast = - f^\ast\,\big(\rho_+^\prime\, \mathcal{G}_{N^\prime}^- + 
\rho_-^\prime\, \mathcal{G}_{N^\prime}^+\big) f_\ast = - \rho_+ \,f^\ast\, \mathcal{G}_{N^\prime}^- \,f_\ast
- \rho_-\, f^\ast \,\mathcal{G}_{N^\prime}^+\,f_\ast=\widetilde{\gamma}_f\quad,
\end{flalign}
where we recall that, by definition, $ f^\ast(\rho^\prime_\pm) =\rho_\pm $.
\sk

To prove the third identity, we use the explicit expressions 
in \eqref{eqn:hfcomponents} and \eqref{eqn:kfcomponents} for the
untilded chain homotopies. 
Because both terms in $\LLL(f) \gamma_f - \lambda_f\,\LLL(f)$ take values in differential forms
supported in the image $f(N)\subseteq N^\prime$ of $f$, we may
prove the third identity by post-composing with $f^\ast$.
The relevant calculation is then given by
\begin{flalign}
\nn f^\ast \,\lambda_f\, \LLL(f) &= - f^\ast\, \rho^\prime_+\, G_{N^\prime}^-\, Q_{N^\prime} \,f_\ast 
-f^\ast\, \rho^\prime_-\, G_{N^\prime}^+ \, Q_{N^\prime}\,f_\ast \\
 &=-  \rho_+\, G_{N}^-\,  Q_{N}
- \rho_-\, G_{N}^+ \, Q_{N} = \gamma_f = f^\ast\, \LLL(f)\, \gamma_f\quad.
\end{flalign}
In the second step we have used $f^\ast(\rho^\prime_\pm) =\rho_\pm $
and naturality of $Q_N$ and $G^\pm_N$.
The last step follows from the fact that $f^\ast \,f_\ast=\id$ on $\LLL(N)$.
\end{proof}

With these preparations, we can now write down an explicit formula
for the zig-zagging chain map \eqref{eqn:zigzagrceexplicit} that models
the relative Cauchy evolution at the level of linear quantum fields.
Because each of the morphisms $f$ in the category $\CC$ in \eqref{eqn:RCEcategory}
is a subset inclusion, we can suppress all occurrences 
of pullbacks $f^\ast$ and pushforwards $f_\ast$ of differential forms 
as these are simply restrictions and extensions by zero.
The explicit formula  is then given by
\begin{flalign}
\nn Z_{(1,M)}^{(0,M)}\,&=\,\big(\id + \partial_{M_{h}}\lambda_{j_-}\big)\, \big(\id + \partial_{M}\lambda_{i_+}\big)\\
\nn \,&=\, \id + \big(\partial_{M_h} \lambda_{j_-}\big)\, \big(\id + \partial_M \lambda_{i_+}\big) + \partial_M \lambda_{i_+}\\
\,&=\, \id + \big((\partial_{M_h}-\partial_M) \lambda_{j_-}\big)\, \big(\id + \partial_M \lambda_{i_+}\big) + \partial_M 
\Big(\lambda_{i_+} + \lambda_{j_-}\,\big(\id + \partial_M \lambda_{i_+}\big)\Big)\quad.\label{eqn:Zexplicit}
\end{flalign}
Because homotopy commutativity of the diagram in Proposition \ref{propo:rcequantumfield} determines
the chain map  $Z_{(1,M)}^{(0,M)}:\LLL(M)\to\LLL(M)$ only up to homotopy, we can drop the
last term in \eqref{eqn:Zexplicit} and consider the chain map
\begin{flalign}\label{eqn:rcelinearexplicit}
\rce_{(M,h)}^{\mathrm{lin}}\,:=\,  \id + \big((\partial_{M_h}-\partial_M) \lambda_{j_-}\big)\, \big(\id + \partial_M \lambda_{i_+}\big)\,:\,
\LLL(M)~\longrightarrow~\LLL(M)
\end{flalign}
as an equivalent model for the relative Cauchy evolution for the linear quantum field.
\begin{propo}\label{propo:rcemapsexplicit}
The chain map \eqref{eqn:rcelinearexplicit} can be simplified as
\begin{flalign}\label{eqn:rcelinearexplicit2}
\rce_{(M,h)}^{\mathrm{lin}} \,=\, \id +  \big(\dd^{\LLL(M_h)} -\dd^{\LLL(M)}\big)
\, G_{M_h}\,Q_{M_h}\, \big(\id + \partial_M \lambda_{i_+}\big)\quad,
\end{flalign}
where $G_{M_h}$ is the causal propagator for the d'Alembert operator on differential forms 
and $Q_{M_h}$ is the differential operator defined in \eqref{eqn:Qmap} on the 
perturbed spacetime $M_h\in\Loc$.
\end{propo}
\begin{proof}
This is a straightforward argument considering the supports of the differential forms involved.
By \eqref{eqn:fNmapping}, we have that $\id + \partial_M \lambda_{i_+} : \LLL(M)\to \LLL(M_+)$ takes
values in the sub-chain complex of linear observables supported in $M_+ \subseteq M$.
We then compute using \eqref{eqn:hfcomponents}
\begin{flalign}
\nn (\partial_{M_h}-\partial_M) \lambda_{j_-}\Big\vert_{\LLL(M_+)} &= \big(\dd^{\LLL(M_h)}-\dd^{\LLL(M)}\big) \,\lambda_{j_-} \Big\vert_{\LLL(M_+)}
+ \lambda_{j_-}\, \big(\dd^{\LLL(M_h)}-\dd^{\LLL(M)}\big) \Big\vert_{\LLL(M_+)}\\
\nn &= \big(\dd^{\LLL(M_h)}-\dd^{\LLL(M)}\big)\, \big(-\rho_+ \,G^-_{M_h} \, Q_{M_h} -\rho_{-}\,G^+_{M_h}\,Q_{M_h}\big)\Big\vert_{\LLL(M_+)}\\
\nn &=  - \big(\dd^{\LLL(M_h)}-\dd^{\LLL(M)}\big)\, G^-_{M_h} \, Q_{M_h} \Big\vert_{\LLL(M_+)}\\
\nn &=  \big(\dd^{\LLL(M_h)}-\dd^{\LLL(M)}\big)\, \big(G^+_{M_h}- G^-_{M_h}\big) \, Q_{M_h} \Big\vert_{\LLL(M_+)}\\
&= \big(\dd^{\LLL(M_h)}-\dd^{\LLL(M)}\big)\, G_{M_h}\,Q_{M_h}\Big\vert_{\LLL(M_+)}\quad,
\end{flalign}
where $\rho_\pm$ denotes the partition of unity arising from the choice of Cauchy surfaces in $M_-$.
In the first step we have written out the mapping complex differentials $\partial_{M_{h}}$ and $\partial_M$.
In the second step we have used that $\big(\dd^{\LLL(M_h)}-\dd^{\LLL(M)}\big) \big\vert_{\LLL(M_+)}=0$
because the differentials $\dd^{\LLL(M_h)}$ and $\dd^{\LLL(M)}$ given in 
\eqref{eqn:linearobscomplex} agree on the complement of the support of the metric perturbation.
The third step uses that $\rho_+=1$ and $\rho_-=0$ on $\supp(h)$ and that the differentials agree elsewhere.
In step four we have 
used that $\supp(h) \cap J^+_{M_h}(M_+)=\emptyset$ and the last step follows from
the definition of the causal propagator.
\end{proof}
\begin{rem}\label{rem:rcemapsexplicit}
Using \eqref{eqn:linearobscomplex} and
\eqref{eqn:Qmap}, we can spell out explicitly the components of the chain map \eqref{eqn:rcelinearexplicit2}.
For a linear ghost field observable $\chi \in\LLL(M)_{-1} = \Omega^0_\cc(M)$, we obtain
\begin{subequations}
\begin{flalign}
\nn \rce^{\mathrm{lin}}_{(M,h)} (\chi) \,&=\,\chi  +  \big(\delta_{M_h} -\delta_M\big) \,G_{M_h}\,\dd_{M_h} \big(\id + \partial_M \lambda_{i_+}\big)(\chi)\\
\,&=\,\chi  +  \big(\square_{M_h} -\square_M\big) \,G_{M_h} \big(\id + \partial_M \lambda_{i_+}\big)(\chi)\quad,
\end{flalign}
where $\square_{M_{(h)}}$ denotes the d'Alembert operator on $0$-forms
on the spacetime $M_{(h)}$. In the second step we have used that $G_{M_h}\,\dd_{M_h} = \dd_{M_h}\,G_{M_h}$ and
$\dd_{M_h}=\dd_M$ for the de Rham differential.
For a linear gauge field observable $\varphi\in\LLL(M)_0 = \Omega^1_\cc(M)$, we obtain
\begin{flalign}\label{eqn:rcevarphi}
\rce^{\mathrm{lin}}_{(M,h)} (\varphi)\,&=\, \varphi + \big(\delta_{M_h}\dd_{M_h} -\delta_M\dd_M\big) \, G_{M_h}\big(\id + \partial_M \lambda_{i_+}\big)(\varphi)\quad.
\end{flalign}
Finally, for linear antifield observables $\alpha\in \LLL(M)_{1} = \Omega^1_\cc(M)$ and
$\beta\in\LLL(M)_{2}=\Omega^0_\cc(M)$, we obtain
\begin{flalign}
\rce^{\mathrm{lin}}_{(M,h)} (\alpha) \,=\, \alpha + \big(\dd_{M_h}-\dd_M\big)\,G_{M_h}\,\delta_{M_h} \big(\id + \partial_M \lambda_{i_+}\big)(\alpha) \,=\,\alpha\quad,
\end{flalign}
because $\dd_{M_h}=\dd_M$, and
\begin{flalign}
\rce^{\mathrm{lin}}_{(M,h)} (\beta) \,=\,\beta\quad, 
\end{flalign}
\end{subequations}
because $Q_{M_h}(\beta) =0$ by \eqref{eqn:Qmap}. Note that our model for the
relative Cauchy evolution is trivial for all linear antifield observables 
$\alpha\in\LLL(M)_{1}$ and $\beta\in \LLL(M)_2$, and it is also trivial in homology
for all linear ghost field observables $\chi\in\LLL(M)_{-1}$, i.e.\ 
\begin{flalign}
\big[\rce^{\mathrm{lin}}_{(M,h)} (\chi) \big] = [\chi] - \big[\square_M \,G_{M_h} \big(\id + \partial_M \lambda_{i_+}\big)(\chi)\big]=[\chi]\quad,
\end{flalign}
where in the first step we have used that $\square_{M_h}G_{M_h}=0$. The second step follows from
\begin{flalign}
\nn \square_M \,G_{M_h} &= \ast_M \dd_M\ast_M \dd_M\,G_{M_h}
= \ast_M \dd_M (\ast_M-\ast_{M_h}) \dd_M\,G_{M_h}\\
&= \delta_M (\id - \ast_M^{-1}\ast_{M_h}) \dd_M\,G_{M_h}\quad,
\end{flalign}
which implies that $\square_M \,G_{M_h} \big(\id + \partial_M \lambda_{i_+}\big)(\chi)$
is exact in the chain complex \eqref{eqn:linearobscomplex}.
Furthermore, when passing to homology and considering 
linear gauge field observables of the form $\varphi = \delta_M\omega\in\LLL(M)_0$ with
$\omega\in \Omega^2_\cc(M_+)$, our expression in 
\eqref{eqn:rcevarphi} agrees with the relative Cauchy evolution 
for the field strength tensor of Maxwell's theory computed in \cite[Section~6.3]{FewsterLang}.
Here, one must take into account the difference in convention regarding the signs of $\Box_M$ and $G_M$.
\end{rem}

We conclude this section by computing the stress-energy tensor of the 
linear Yang-Mills model. To simplify the resulting expressions,
we pre-compose the chain map \eqref{eqn:rcelinearexplicit2}
with the quasi-isomorphism $\LLL(i_+) : \LLL(M_+)\to \LLL(M)$
and observe that the resulting chain map is homotopic to
\begin{flalign}
\rce_{(M,h)}^{\mathrm{lin},+} \,:=\, \id +  \big(\dd^{\LLL(M_h)} -\dd^{\LLL(M)}\big)
\, G_{M_h}\,Q_{M_h}\,:\, \LLL(M_+)~\longrightarrow~\LLL(M)\quad.
\end{flalign}
In terms of the components displayed in Remark \ref{rem:rcemapsexplicit}, this means
that we restrict to differential forms with compact support in $M_+$
and drop the chain map $\id + \partial_M \lambda_{i_+}$.
We then compute the first derivative
\begin{subequations}
\begin{flalign}
t_{(M,h)}\,:=\, \frac{d}{d\epsilon}\rce_{(M,\epsilon h)}^{\mathrm{lin},+}\big\vert_{\epsilon=0} \,:\, \LLL(M_+)~\longrightarrow~\LLL(M)\quad,
\end{flalign}
which is a chain map whose components read as
\begin{flalign}
\nn t_{(M,h)}(\chi) \,&=\,{\nabla_{\! M}}_a\big(h^{ab} \, (\dd_M G_M \chi)_b\big)-
\tfrac{1}{2} \big({\nabla_{\! M}}_b h^{a}_{~a}\big) \,(\dd_M G_M \chi)^b\quad,\\
\nn t_{(M,h)}(\varphi)_c  \,&=\, {\nabla_{\! M}}_a\big(h^{ab} \, (\dd_M G_M \varphi)_{bc}\big)-
\tfrac{1}{2} \big({\nabla_{\! M}}_b h^{a}_{~a}\big) \,(\dd_M G_M \varphi)^b_{~c} +
\big({\nabla_{\! M}}_a h_{bc}\big) \,(\dd_M G_M \varphi)^{ab}\quad,\\
\nn t_{(M,h)}(\alpha)_c \,&=\, 0\quad,\\
t_{(M,h)}(\beta) \,&=\, 0 \quad,
\end{flalign}
\end{subequations}
for all $\chi\in \LLL(M_+)_{-1}$, $\varphi\in\LLL(M_+)_0$,
$\alpha\in\LLL(M_+)_1$ and $\beta\in\LLL(M_+)_2$.
Here we have used as in \cite{FewsterLang} an index notation for tensor fields 
and differential forms on $M$.  The Levi-Civita connection $\nabla_{\! M}$ on the spacetime 
$M\in\Loc$ enters these expressions through the following identity
\begin{multline}
(\delta_{M_{\epsilon h}}\omega  -\delta_M\omega)_{a_1\cdots a_k} 
\,=\, \epsilon\,\bigg({\nabla_{\! M}}_a \big(h^{ab}\,\omega_{b a_1\cdots a_k}\big)
-\tfrac{1}{2}  \big({\nabla_{\! M}}_b h^{a}_{~a}\big)\,\omega^b_{~a_1\cdots a_k}\\
+ \sum_{j=1}^k (-1)^{j-1} \, \big({\nabla_{\! M}}_a  h_{ba_j}\big)\,\omega^{ab}_{~~a_1\cdots \widehat{a_j}\cdots a_k}\bigg)
+\mathcal{O}(\epsilon^2)\quad,
\end{multline}
for all $k+1$-forms $\omega = 
\omega_{a_0 a_1\cdots a_k}dx^{a_0}\wedge dx^{a_1}\wedge \cdots\wedge dx^{a_k}\in \Omega^{k+1} (M)$.
\sk

In contrast to the situation in ordinary AQFT, the chain map
$t_{(M,h)} : \LLL(M_+)\to \LLL(M)$ does {\em not} extend to a
derivation $\AAA(M_+)=\CCR(\LLL(M_+),\tau_{M_+})\to \AAA(M)=\CCR(\LLL(M),\tau_M)$ 
relative to $\AAA(i_+) : \AAA(M_+)\to \AAA(M)$  at the level of the observable dg-algebras.
This is a remnant of the fact that $\rce_{(M,h)}^{\mathrm{lin},+}$ does not preserve the Poisson structures,
which at the infinitesimal level (i.e.\ to first order in $\epsilon$) amounts to
\begin{flalign}
\tau_M\circ (t_{(M,h)}\wedge \LLL(i_+)) + \tau_{M}\circ (\LLL(i_+)\wedge t_{(M,h)}) \, \neq\, 0
\end{flalign}
as a chain map $\LLL(M_+)\wedge \LLL(M_+)\to \bbR$.
This issue can be rectified by considering the
homotopic model for $t_{(M,h)}$ given by the chain map
\begin{subequations}
\begin{flalign}
\widetilde{t}_{(M,h)} \,:=\, t_{(M,h)} + \partial \psi\,:\, \LLL(M_+)~\longrightarrow~\LLL(M)
\end{flalign}
and the chain homotopy $\psi\in\hom(\LLL(M_+),\LLL(M))_1$ 
defined by the components
\begin{flalign}
\psi(\chi)_b \,&=\, \tfrac{1}{2} h^{a}_{~a} \,(\dd_M G_M \chi)_b - h_{ab}\,(\dd_M G_M\chi)^a\quad,\\
\psi(\varphi)_b\,&=0\,~~,\quad \psi(\alpha)\,=\, 0~~,\quad \psi(\beta)\,=\,0\quad.
\end{flalign}
\end{subequations}
By a straightforward calculation using integration by parts, one proves that 
\begin{flalign}
\tau_M\circ (\widetilde{t}_{(M,h)}\wedge \LLL(i_+)) + \tau_{M}\circ (\LLL(i_+)\wedge \widetilde{t}_{(M,h)}) \,=\,0\quad,
\end{flalign}
which implies that $\widetilde{t}_{(M,h)}$ extends to derivation  $\widetilde{t}_{(M,h)} : 
\AAA(M_+)\to \AAA(M)$ relative to $\AAA(i_+) : \AAA(M_+)\to \AAA(M)$.
Following the arguments in \cite{FewsterLang} (and again adjusting for a 
difference in sign conventions for $G_M$ and therefore the Poisson structure), 
this derivation determines a polarized form of the stress-energy tensor through
\begin{flalign}
\tau_M\big(\widetilde{t}_{(M,h)}(\omega_1)\otimes\LLL(i_+)(\omega_2)\big)
\,=:\,  \int_M h_{ab}\,T^{ab}_M(\omega_1,\omega_2)\,\vol_M\quad,
\end{flalign}
for all $\omega_1,\omega_2\in\LLL(M_+)$. By another straightforward calculation
using integration by parts, we find the explicit formula
\begin{flalign}
T^{ab}_M(\omega_1,\omega_2)\,=\,
\tfrac{1}{4} g^{ab}\, (F_{\varphi_1})^{cd}\,(F_{\varphi_2})_{cd} -(F_{\varphi_1})^{ac}\, (F_{\varphi_2})^{b}_{~c}\quad,
\end{flalign}
where $\varphi_i :=\mathrm{pr}_0(\omega_i)\in\LLL(M_+)_0$ denotes 
the degree $0$ component of $\omega_i$ and $F_{\varphi_i} := \dd_M G_M\varphi_i \in\Omega^2(M)$
denotes its associated field strength $2$-form.
Observe that our polarized stress-energy tensor does not receive contributions from the ghosts and the
antifields, and that in degree $0$ it agrees with the usual Maxwell stress-energy tensor 
(in our $(+-\cdots-)$ signature) on setting $\omega_2=\omega_1$. This result shows how these
physically expected results can be provided with a first-principles justification on quite abstract grounds.


\section*{Acknowledgments}
We would like to thank Marco Benini and Victor Carmona for useful comments on this work.
S.B.\ is supported by a PhD scholarship (RG160517) of the Royal Society (UK).
A.S.\ gratefully acknowledges the financial support of 
the Royal Society (UK) through a Royal Society University 
Research Fellowship (UF150099), a Research Grant (RG160517) 
and two Enhancement Awards (RGF\textbackslash EA\textbackslash 180270 
and RGF\textbackslash EA\textbackslash 201051).

\appendix

\section{\label{app:bimcomplexes}Conventions for bicomplexes}
A {\em bicomplex} $V$ is a bigraded family of $\bbK$-vector spaces
$\{V_{p,q}\}_{p,q\in\bbZ}$ equipped with a vertical
differential $\delta : V_{p,q}\to V_{p-1,q}$ and a horizontal
differential $\dd : V_{p,q}\to V_{p,q-1}$ satisfying
\begin{flalign}
\delta^2 =0~~,\quad \dd^2=0~~,\quad \delta\,\dd + \dd\,\delta =0\quad.
\end{flalign}
Observe that the vertical and horizontal differentials are required to {\em anti-commute},
which is more convenient for discussing the symmetric monoidal structure and totalization 
of bicomplexes than the alternative convention of commuting differentials. We refer to e.g.\ \cite{Bicomplexes}
for more details and the relevant argument that both conventions are equivalent.
A morphism of bicomplexes $f : V\to W$ is a family of linear maps 
$\{f_{p,q} : V_{p,q}\to W_{p,q}\}_{p,q\in\bbZ}$ that commutes with both the vertical 
and the horizontal differential.
We denote by $\bCh_\bbK$ the category of bicomplexes of $\bbK$-vector spaces.
\sk

The category $\bCh_\bbK$ is symmetric monoidal with respect to the tensor 
product defined by
\begin{subequations}
\begin{flalign}
(V\otimes W)_{p,q} \,=\, \bigoplus_{\mycom{i+k=p}{j+l=q}} V_{i,j} \otimes W_{k,l}
\end{flalign}
together with the differentials
\begin{flalign}
\delta(v\otimes w) \,&=\, \delta(v) \otimes w + (-1)^{\vert v\vert^{\mathrm{tot}}}\, v\otimes \delta(w)\quad,\\
\dd(v\otimes w) \,&=\, \dd(v)\otimes w + (-1)^{\vert v \vert^{\mathrm{tot}}}\, v\otimes\dd(w)\quad,
\end{flalign}
\end{subequations}
where $\vert v\vert^{\mathrm{tot}} = i+j$  denotes the {\em total degree} of $v\in V_{i,j}$.
The monoidal unit is $\bbK$ concentrated in degree $(0,0)$
and the symmetric braiding is given by the Koszul sign rule with respect
to the total degrees, i.e.\
\begin{flalign}
V\otimes W ~\longrightarrow~W\otimes V~,~~v\otimes w ~\longmapsto~(-1)^{\vert v\vert^{\mathrm{tot}}\,\vert w\vert^{\mathrm{tot}}}\,w\otimes v\quad.
\end{flalign}
The {\em $\bigoplus$-totalization} of bicomplexes is given by the functor
\begin{flalign}
\Tot \,:\, \bCh_\bbK~\longrightarrow~\Ch_\bbK
\end{flalign}
that assigns to a bicomplex $V\in\bCh_\bbK$ the chain complex defined by
\begin{subequations}
\begin{flalign}
\Tot(V)_m\,:=\, \bigoplus_{i+j=m} V_{i,j}
\end{flalign}
together with the differential
\begin{flalign}
\dd^{\mathrm{tot}}\,:=\, \delta + \dd\quad.
\end{flalign}
\end{subequations}
It is easy to check that $\Tot$ is a strong symmetric monoidal functor
with respect to the obvious structure maps.

\section{\label{app:bar}Bar construction}
The bar construction is a powerful and efficient tool to obtain derived functors.
In this appendix we shall spell out some computational details
that will help the reader to understand better our explicit formulas in Section \ref{sec:rectification}.
We refer to \cite{Fresse} and also \cite{Riehl} for a more detailed presentation of the bar construction.
\sk

Let $F : \CC\to\DD$ be a functor between any two categories $\CC$ and $\DD$.
(The case of interest in the bulk of the paper is the localization functor
$L:\CC\to \BB\bbZ$ from \eqref{eqn:localizationfunctor}.)
The associated bar construction is then a functor
\begin{flalign}
B_{\Delta}(\DD,\CC,-)\,:\, \Fun(\CC,\Ch_\bbK)~\longrightarrow ~\Fun(\DD,\mathbf{sCh}_\bbK)
\end{flalign}
that assigns to each chain complex-valued functor $X : \CC\to\Ch_{\bbK}$ on $\CC$ a functor 
$B_{\Delta}(\DD,\CC,X) : \DD\to \mathbf{sCh}_\bbK$ on $\DD$ with values in simplicial chain complexes.
In order to avoid simplicial technology, which may be unfamiliar to some readers,
we prefer to present this construction in the more familiar language of bicomplexes, 
see Appendix \ref{app:bimcomplexes} for our conventions. (Technically, this uses the Dold-Kan correspondence
between simplicial chain complexes and bicomplexes.) From this perspective the bar construction is a functor
\begin{flalign}
\overline{B_{\Delta}}(\DD,\CC,-)\,:\, \Fun(\CC,\Ch_\bbK)~\longrightarrow ~\Fun(\DD,\bCh_\bbK)
\end{flalign}
assigning functors on $\DD$ with values in bicomplexes.
The application of this functor on any $X\in\Fun(\CC,\Ch_\bbK)$ admits the following
explicit description: The functor $\overline{B_{\Delta}}(\DD,\CC,X):\DD\to \bCh_\bbK$
assigns to an object $d\in\DD$ the bicomplex concentrated in vertical degrees $m\geq 0$
given at vertical degree $m=0$ by
\begin{subequations}\label{eqn:bargeneral}
\begin{flalign}
\overline{B_{\Delta}}(\DD,\CC,X)(d)_{0,\bullet}\,=\,\bigoplus_{c\in\CC} ~\bigoplus_{g\in \DD(Fc,d)}~X(c)_\bullet\quad
\end{flalign}
and for vertical degrees $m\ge 1$ by
\begin{flalign}\label{eqn:bargeneralmge1}
\overline{B_{\Delta}}(\DD,\CC,X)(d)_{m,\bullet}\,=\,\bigoplus_{c\in\CC} ~\bigoplus_{g\in \DD(Fc,d)}~\bigoplus_{\mycom{(f_1,\dots,f_m)\in \mathrm{Mor}_m(\CC)}{\mathrm{t}f_1 =c\,,~f_i\neq \id}} ~X(\mathsf{s}f_m)_\bullet\quad,
\end{flalign}
where $\mathrm{Mor}_m(\CC)$ denotes the set of
composable $m$-tuples 
$\mathsf{t}f_1  \stackrel{f_1}{\longleftarrow} \mathsf{s}f_1=\mathsf{t}f_2 \stackrel{f_2}{\longleftarrow}
 \cdots \stackrel{f_m}{\longleftarrow} \mathsf{s}f_m$ 
of morphisms in $\CC$ and $\mathsf{s}$/$\mathsf{t}$ denotes the source/target 
of a morphism. (Note that $g :Fc\to d$ is a $\DD$-morphism, while the $f_i$ are morphisms in $\CC$.)
The vertical differential is given by the alternating sum
\begin{flalign}
\nn \delta\big(c,g,f_1,\dots,f_m,x\big)\,&=\,(-1)^{\vert x\vert}\, \Big(
\big(\mathsf{s}f_1,g\circ Ff_1,f_2,\dots,f_m,x\big) \\
\nn &\qquad\quad + \sum_{j=1}^{m-1}(-1)^j ~\big(c,g,f_1,\dots,f_j\circ f_{j+1},\dots,f_m,x\big)\\
&\qquad\quad + (-1)^m\, \big(c,g,f_1,\dots,f_{m-1}, X(f_m)x\big )
\Big)\quad,
\end{flalign}
\end{subequations}
for all $(c,g,f_1,\dots,f_m,x)\in \overline{B_{\Delta}}(\DD,\CC,X)(d)_{m,\bullet}$,
where $\vert x\vert$ denotes the degree of $x\in X(\mathsf{s}f_m)$.
To any $\DD$-morphism $k : d\to d^\prime$, the functor assigns the map of bicomplexes
\begin{flalign}
\nn\overline{B_{\Delta}}(\DD,\CC,X)(k) \,:\, \overline{B_{\Delta}}(\DD,\CC,X)(d)~&\longrightarrow~\overline{B_{\Delta}}(\DD,\CC,X)(d^\prime)\quad,\\
\big(c,g,f_1,\dots,f_m,x\big) ~&\longmapsto~\big(c,k\circ g,f_1,\dots,f_m,x\big)
\end{flalign}
determined by post-composition with $k$.
\sk

Roughly speaking, the role of the bar construction is to `fatten up'
the ordinary left Kan extension $F_! = \Lan_F : \Fun(\CC,\Ch_\bbK)\to \Fun(\DD,\Ch_\bbK)$
to provide a homotopically meaningful construction, i.e.\ one that is compatible 
with weak equivalences. (Note that the $0$-th vertical homology
of $\overline{B_{\Delta}}(\DD,\CC,X)$ is a model for $F_!X$.) This is done by
introducing many redundant copies via the direct sums over composable tuples 
of arrows in \eqref{eqn:bargeneral}. This allows us to define a model for the derived
left Kan extension 
\begin{flalign}
\bbL F_!\,:=\, \Tot\big(\overline{B_{\Delta}}(\DD,\CC,-) \big)\,:\, 
\Fun(\CC,\Ch_\bbK)~\longrightarrow ~\Fun(\DD,\Ch_\bbK)
\end{flalign}
by applying the totalization functor $\Tot$ from Appendix \ref{app:bimcomplexes}.
For a proof of this statement (and its generalization to operad algebras) 
we refer to e.g.\ \cite[Theorem 17.2.7 and Section 13.3]{Fresse}. 
\sk

We shall also need concrete models for the derived counit and unit associated with
the derived left Kan extension. Denoting by $F^\ast : \Fun(\DD,\Ch_\bbK)\to\Fun(\CC,\Ch_\bbK)$
the pullback functor, which we recall does not have to be derived because each object in 
$\Fun(\DD,\Ch_\bbK)$ is fibrant in the projective model structure, the component
at $Y\in \Fun(\DD,\Ch_\bbK)$ of the derived counit is the natural transformation
\begin{subequations}
\begin{flalign}
\epsilon_Y\,:\, \bbL F_! F^\ast(Y) ~\longrightarrow~ Y
\end{flalign}
defined by the following components: For each $d\in\DD$, 
\begin{flalign}
\nn \epsilon_{Y,d}\,:\,\bbL F_! F^\ast(Y)(d) ~&\longrightarrow~ Y(d)\quad,\\
\nn (c,g,y)~&\longmapsto~ Y(g)y\quad,\\
(c,g,f_1,\dots,f_m,y)~&\longmapsto~0\quad,
\end{flalign}
\end{subequations}
for all vertical degrees $m\geq 1$.
\sk

The derived unit is slightly more involved to describe. Let us denote by
\begin{flalign}
Q \,:=\, \Tot\big(\overline{B_\Delta}(\CC,\CC,-)\big)\,:\, \Fun(\CC,\Ch_\bbK)~\longrightarrow~\Fun(\CC,\Ch_\bbK)
\end{flalign}
the resolution that is obtain by totalizing the bar construction for the identity functor $\id:\CC\to\CC$.
The component at $X\in \Fun(\CC,\Ch_\bbK)$ of the derived unit
is the natural transformation
\begin{subequations}
\begin{flalign}
\eta_{X}\,:\, Q(X) ~\longrightarrow~ F^\ast \bbL F_!(X)
\end{flalign}
defined by the following components: For each $\tilde{c}\in \CC$,
\begin{flalign}
\nn \eta_{X,\tilde{c}} \,:\, Q(X)(\tilde{c})~&\longrightarrow~\bbL F_!(X)(F\tilde{c})\quad,\\
(c,f_0,f_1,\dots, f_m,x)~&\longmapsto~(c,Ff_0,f_1,\dots,f_m,x)\quad,
\end{flalign}
\end{subequations}
for all $m\geq 0$, where we note that $f_0: c\to\tilde{c}$ corresponds to
$g: Fc\to d$ in \eqref{eqn:bargeneral} because $Q(X)$ is obtained from the bar
construction associated with the identity functor $\id:\CC\to \CC$.

\section{\label{app:zigzag}Zig-zagging homotopy coherence data}
The zig-zagging homotopies and $2$-homotopies  postulated in
\eqref{eqn:zigzagginghomotopyh}, \eqref{eqn:zigzagginghomotopyk} 
and \eqref{eqn:zigzagging2homotopy} can be constructed
explicitly from the homotopy data \eqref{eqn:homotopy4quasiinverse} 
associated to our choice of quasi-inverses $\LLL(f)^{-1}$.
This construction requires some case distinctions,
for which we exclude for the moment the case of identity morphisms.
(We shall see below that the homotopy coherence data for identity 
morphisms can be set consistently to $0$.)
\sk

Given $(n^\prime,N^\prime)\in \bbZ\times \CC$
and $(n,f)\in \bbZ\times \mathrm{Mor}(\CC)$,
the arrow $f$ can point either {\em along} the direction of
the shortest zig-zag in \eqref{eqn:spiral} from $(n^\prime,N^\prime)$ 
to $ (n+L(f),\mathsf{s}f)$, i.e.\
\begin{subequations}
\begin{flalign}\label{eqn:caseA1}
\xymatrix@C=3em{
(n^\prime,N^\prime)\ar@{<->}[r]^-{\text{zig-zag}}~&~ (n+L(f),\mathsf{s}f)\ar[r]^-{f} ~&~ (n,\mathsf{t}f)
}\quad,
\end{flalign}
or it can point {\em against} the shortest zig-zag from $(n^\prime,N^\prime)$ 
to $(n,\mathsf{t}f)$, i.e.\
\begin{flalign}\label{eqn:caseA2}
\xymatrix@C=3em{
(n^\prime,N^\prime)\ar@{<->}[r]^-{\text{zig-zag}}~&~ (n,\mathsf{t}f)~&~ \ar[l]_-{f}  (n+L(f),\mathsf{s}f)
}\quad.
\end{flalign}
\end{subequations}
Note that the orientation in these pictures is given by the direction of the zig-zag 
and hence it does not have to coincide with the orientation of \eqref{eqn:spiral}.
We also note that the situation $(n^\prime,N^\prime) = (n+L(f),\mathsf{s}f)$
is covered by the first case, while $(n^\prime,N^\prime) = (n,\mathsf{t}f)$
is covered by the second one. 
\sk

Similarly, given $(n^\prime,f^\prime)\in \bbZ\times \mathrm{Mor}(\CC)$
and $(n,N)\in \bbZ\times \CC$,
the arrow $f^\prime$ can point either along the direction of
the shortest zig-zag in \eqref{eqn:spiral} to $(n,N)$, i.e.\
\begin{subequations}
\begin{flalign}\label{eqn:caseB1}
\xymatrix@C=3em{
(n^\prime+L(f^\prime),\mathsf{s}f^\prime) \ar[r]^-{f^\prime} ~&~ (n^\prime,\mathsf{t} f^\prime)
\ar@{<->}[r]^-{\text{zig-zag}} ~&~ (n,N)
}\quad,
\end{flalign}
or it can point against it, i.e.\
\begin{flalign}\label{eqn:caseB2}
\xymatrix@C=3em{
(n^\prime,\mathsf{t}f^\prime) ~&~\ar[l]_-{f^\prime}  (n^\prime+L(f^\prime),\mathsf{s} f^\prime)
\ar@{<->}[r]^-{\text{zig-zag}} ~&~ (n,N)
}\quad.
\end{flalign}
\end{subequations}

Given two arrows $(n,f),(n^\prime,f^\prime)\in \bbZ\times \mathrm{Mor}(\CC)$,
there exist $5$ different cases for their alignment relative to the zig-zag from $(n^\prime,f^\prime)$ 
to $(n,f)$, namely
\begin{subequations}
\begin{flalign}\label{eqn:caseC1}
\xymatrix@C=3em{
(n^\prime+L(f^\prime),\mathsf{s}f^\prime) \ar[r]^-{f^\prime} ~&~ (n^\prime,\mathsf{t} f^\prime)\ar@{<->}[r]^-{\text{zig-zag}}
~&~  (n,\mathsf{t}f)~&~ \ar[l]_-{f}  (n+L(f),\mathsf{s}f)
}\quad,
\end{flalign}
\begin{flalign}\label{eqn:caseC2}
\xymatrix@C=3em{
(n^\prime,\mathsf{t}f^\prime)~&~ \ar[l]_-{f^\prime}  (n^\prime+L(f^\prime),\mathsf{s} f^\prime)\ar@{<->}[r]^-{\text{zig-zag}}
~&~  (n+L(f),\mathsf{s}f) \ar[r]^-{f}~&~  (n,\mathsf{t}f)
}\quad,
\end{flalign}
\begin{flalign}\label{eqn:caseC3}
\xymatrix@C=3em{
(n^\prime+L(f^\prime),\mathsf{s}f^\prime) \ar[r]^-{f^\prime} ~&~ (n^\prime,\mathsf{t} f^\prime)\ar@{<->}[r]^-{\text{zig-zag}}
~&~   (n+L(f),\mathsf{s}f) \ar[r]^-{f}~&~  (n,\mathsf{t}f)
}\quad,
\end{flalign}
\begin{flalign}\label{eqn:caseC4}
\xymatrix@C=3em{
(n^\prime,\mathsf{t}f^\prime)~&~ \ar[l]_-{f^\prime}  (n^\prime+L(f^\prime),\mathsf{s} f^\prime)\ar@{<->}[r]^-{\text{zig-zag}}
~&~   (n,\mathsf{t}f)~&~ \ar[l]_-{f}  (n+L(f),\mathsf{s}f)
}\quad,
\end{flalign}
\begin{flalign}\label{eqn:caseC5}
(n,f) = (n^\prime,f^\prime)\quad.
\end{flalign}
\end{subequations}

Using these case distinctions, we define the zig-zagging homotopy (for left composition) in \eqref{eqn:zigzagginghomotopyh} by
\begin{flalign}\label{eqn:Hhformula}
\Lambda_{(n^\prime,N^\prime)}^{(n,f)}\,=\,\begin{cases}
0 &,~~\text{for case \eqref{eqn:caseA1}}\quad,\\
\lambda_f\,Z_{(n^\prime,N)}^{(n,\mathsf{t}f)} &,~~\text{for case \eqref{eqn:caseA2}}\quad.
\end{cases}
\end{flalign}
The zig-zagging homotopy (for right composition) in \eqref{eqn:zigzagginghomotopyk} is given by
\begin{flalign}\label{eqn:Hkformula}
\Gamma_{(n^\prime,f^\prime)}^{(n,N)}\,=\,\begin{cases}
0 & ,~~\text{for case \eqref{eqn:caseB1}}\quad,\\
Z^{(n,N)}_{(n^\prime +L(f^\prime),\mathsf{s}f^\prime)}\,\gamma_{f^\prime}
& ,~~\text{for case \eqref{eqn:caseB2}}\quad,
\end{cases}
\end{flalign}
and the zig-zagging $2$-homotopy in \eqref{eqn:zigzagging2homotopy} is given by
\begin{flalign}\label{eqn:Kformula}
\Xi_{(n^\prime,f^\prime)}^{(n,f)}\,=\,
\begin{cases}
0 &,~~\text{for cases \eqref{eqn:caseC1}, \eqref{eqn:caseC2} and \eqref{eqn:caseC3}}\quad,\\
\lambda_f\,Z_{(n^\prime+L(f^\prime),\mathsf{s}f^\prime)}^{(n,\mathsf{t}f)}\gamma_{f^\prime} &,~~\text{for case \eqref{eqn:caseC4}}\quad,\\
\xi_f &,~~\text{for case \eqref{eqn:caseC5}}\quad.
\end{cases}
\end{flalign}
Observe that the case distinctions in these expressions all collapse to $0$ 
in the case where $f$ or $f^\prime$ is an identity morphism
since we have chosen $\lambda_\id =0$, $\gamma_\id=0$ and $\xi_\id=0$. Hence, we can set consistently
the zig-zagging homotopies associated with an identity morphism
and also the zig-zagging $2$-homotopies associated with at least one identity morphism
to zero.
\sk

It is easy to confirm that \eqref{eqn:Hhformula}, \eqref{eqn:Hkformula}
and \eqref{eqn:Kformula} provides the required homotopy coherence data in 
\eqref{eqn:zigzagginghomotopyh}, \eqref{eqn:zigzagginghomotopyk} and \eqref{eqn:zigzagging2homotopy}.
As an illustration, consider for example \eqref{eqn:zigzagginghomotopyh} for the case \eqref{eqn:caseA2}.
Then we compute
\begin{flalign}
\resizebox{.9 \textwidth}{!} 
{$
\LLL(f)\, Z^{(n+L(f),\mathsf{s}f)}_{(n^\prime,N^\prime)} - Z^{(n,\mathsf{t}f)}_{(n^\prime,N^\prime)} \,=\,
\big(\LLL(f)\, \LLL(f)^{-1} -\id \big) \,Z^{(n,\mathsf{t}f)}_{(n^\prime,N^\prime)} \,=\,\partial\big(\lambda_f\,Z^{(n,\mathsf{t}f)}_{(n^\prime,N^\prime)}\big)
\,=\, \partial \Lambda^{(n,f)}_{(n^\prime,N^\prime)}~~.
$}
\end{flalign}
The other cases follow similarly.



\begin{thebibliography}{10}

\bibitem[B\"ar15]{Baer}
C.~B\"ar,
``Green-Hyperbolic operators on globally hyperbolic spacetimes,''
Commun.\ Math.\ Phys.\ \textbf{333}, no.\ 3, 1585--1615 (2015)
[arXiv:1310.0738 [math-ph]].


\bibitem[BBS19]{LinearYM}
M.~Benini, S.~Bruinsma and A.~Schenkel,
``Linear Yang-Mills theory as a homotopy AQFT,''
Commun.\ Math.\ Phys.\ \textbf{378}, no.\ 1, 185--218 (2019)
[arXiv:1906.00999 [math-ph]].


\bibitem[BS19a]{BSreview} 
M.~Benini and A.~Schenkel,
``Higher structures in algebraic quantum field theory,''
Fortsch.\ Phys.\  {\bf 67}, no.\ 8-9, 1910015 (2019)
[arXiv:1903.02878 [hep-th]].
  
  
\bibitem[BSW21]{BSWoperad} 
M.~Benini, A.~Schenkel and L.~Woike,
``Operads for algebraic quantum field theory,''
Commun. \ Contemp. \ Math. \ {\bf 23},  no.\ 2, 2050007 (2021)
[arXiv:1709.08657 [math-ph]].

  
\bibitem[BSW19]{BSWhomotopy}
M.~Benini, A.~Schenkel and L.~Woike, 
``Homotopy theory of algebraic quantum field theories,''
Lett.\ Math.\ Phys.\  {\bf 109}, no.\ 7, 1487 (2019)
[arXiv:1805.08795 [math-ph]]. 


\bibitem[BS19b]{BruinsmaSchenkel}
S.~Bruinsma and A.~Schenkel,
``Algebraic field theory operads and linear quantization,''
Lett.\ Math.\ Phys.\ \textbf{109}, no.\ 11, 2531--2570 (2019)
[arXiv:1809.05319 [math-ph]].


\bibitem[BFV03]{BFV} 
R.~Brunetti, K.~Fredenhagen and R.~Verch,
``The generally covariant locality principle: A new paradigm for local quantum field theory,''
Commun.\ Math.\ Phys.\ {\bf 237}, 31 (2003)
[math-ph/0112041].


\bibitem[Car21]{Carmona}
V.~Carmona,
``Algebraic Quantum Field Theories: a homotopical view,''
arXiv:2107.14176 [math-ph].


\bibitem[CG17]{CostelloGwilliam}
K. ~Costello and O.~Gwilliam, 
{\it Factorization algebras in quantum field theory},
New Mathematical Monographs {\bf 31}, 
Cambridge University Press,  Cambridge (2017).


\bibitem[DHKS04]{DHKS}
W.~G.~Dwyer, P.~S.~Hirschhorn, D.~M.~Kan and J.~H.~Smith,
{\it Homotopy limit functors on model categories and homotopical categories},
Math.\ Surveys Monogr.\ {\bf 113}, 
Amer.\ Math.\ Soc., Providence, RI (2004).


\bibitem[FL16]{FewsterLang}
C.~J.~Fewster and B.~Lang,
``Dynamical locality of the free Maxwell field,''
Annales Henri Poincar{\'e} \textbf{17}, no.\ 2, 401--436 (2016)
[arXiv:1403.7083 [math-ph]].


\bibitem[FV12]{FewsterVerchSPASS}
C.~J.~Fewster and R.~Verch,
``Dynamical locality and covariance: What makes a physical theory the same in all spacetimes?,''
Annales Henri Poincar{\'e} \textbf{13}, 1613--1674 (2012)
[arXiv:1106.4785 [math-ph]].


\bibitem[FV15]{FewsterVerch}
C.~J.~Fewster and R.~Verch,
``Algebraic quantum field theory in curved spacetimes,''
in: R.~Brunetti, C.~Dappiaggi, K.~Fredenhagen and J.~Yngvason (eds.),
{\it Advances in algebraic quantum field theory}, 
Springer Verlag, Heidelberg (2015)
[arXiv:1504.00586 [math-ph]].


\bibitem[FR12]{FredenhagenRejzner} 
K.~Fredenhagen and K.~Rejzner,
``Batalin-Vilkovisky formalism in the functional approach to classical field theory,''
Commun.\ Math.\ Phys.\  {\bf 314}, 93 (2012)
[arXiv:1101.5112 [math-ph]].
  
  
\bibitem[FR13]{FredenhagenRejzner2} 
K.~Fredenhagen and K.~Rejzner,
``Batalin-Vilkovisky formalism in perturbative algebraic quantum field theory,''
Commun.\ Math.\ Phys.\  {\bf 317}, 697 (2013)
[arXiv:1110.5232 [math-ph]].


\bibitem[Fre09]{Fresse}
B.~Fresse,
{\it Modules over operads and functors},
Lecture Notes in Mathematics {\bf 1967},
Springer Verlag, Berlin Heidelberg (2009).


\bibitem[Hov99]{Hovey}
M.~Hovey, 
{\it Model categories}, 
Math.\ Surveys Monogr.\ {\bf 63}, 
Amer.\ Math.\ Soc., Providence, RI (1999).


\bibitem[KS06]{KashiwaraSchapira}
M.~Kashiwara and P.~Schapira, 
{\it Categories and sheaves}, 
Springer Verlag, Berlin (2006). 


\bibitem[Lac10]{Lack}
S.~Lack,
``A $2$-categories companion,'' 
in:  J.~C.~Baez and J.~P.~May (eds.), 
{\it Towards higher categories}, 
IMA Vol.\ Math.\ Appl.\ {\bf 152}, 105--191, 
Springer Verlag, New York (2010)
[arXiv:math/0702535 [math.CT]]. 


\bibitem[LV12]{LodayVallette}
J.-L.~Loday and B.~Vallette,
{\it Algebraic operads},
Grundlehren der Mathematischen Wissenschaften {\bf 346},
Springer Verlag, Heidelberg (2012).


\bibitem[MR19]{Bicomplexes}
F.~Muro and C.~Roitzheim,
``Homotopy theory of bicomplexes,''
Journal of Pure and Applied Algebra {\bf 223},  1913--1939 (2019)
[arXiv:1802.07610 [math.AT]].


\bibitem[Rie14]{Riehl}
E.~Riehl,
{\it Categorical homotopy theory},
New Mathematical Monographs {\bf 24},
Cambridge University Press, Cambridge (2014).


\bibitem[Wei94]{Weibel}
C.~A.~Weibel, 
{\it An introduction to homological algebra}, 
Cambridge Studies in Advanced Mathematics {\bf 38}, 
Cambridge University Press, Cambridge (1994).

\end{thebibliography}
\end{document}